\documentclass[journal,twoside,web]{IEEEtran}

\usepackage{cite}

\usepackage[colorlinks=true,linkcolor=black]{hyperref}
\usepackage{url}

\usepackage{textcomp} 

\usepackage{bm} 

\usepackage{amsmath,amssymb,amsfonts}
\usepackage{mathtools} 
\usepackage{amsthm} 
\usepackage{commath} 
\usepackage[bb=dsserif]{mathalpha} 

\usepackage{siunitx}

\usepackage{graphicx}

\usepackage{tabularx,ragged2e}
\newcolumntype{L}{>{\raggedright\arraybackslash}X} 
\newcolumntype{C}{>{\Centering\arraybackslash}X} 
\newcolumntype{R}{>{\raggedleft\arraybackslash}X} 
\usepackage{multirow}

\usepackage[ruled,linesnumbered]{algorithm2e}
\usepackage{algpseudocode}
\SetKwProg{Function}{Function}{}{end}
\SetKwRepeat{Do}{do}{while}%
\SetKwInOut{Input}{Input}
\SetKwInOut{Output}{Output}
\SetKwInOut{Parameter}{Param}

\SetCommentSty{mycommentfont}
\SetKwComment{Comment}{$\triangleright$\ }{}
\SetKwFunction{optimalIncentive}{optimalIncentive}%
\SetKwFunction{LoMPCSolution}{LoMPCSolution}%
\SetKwFunction{incentiveUpdate}{updateIncentive}%
\SetKwFunction{incentiveRegularization}{regularizeIncentive}%
\SetKwFunction{solveIncentiveBiMPC}{solveIncentiveBiMPC}%
\SetKwFunction{solveRobustBiMPC}{solveRobustBiMPC}%

\usepackage{float}
\usepackage[font={small}]{caption}
\usepackage{subcaption}

\usepackage{enumitem}

\usepackage{datetime}

\usepackage{balance} 

\newtheorem{proposition}{Proposition}
\newtheorem{theorem}{Theorem}
\newtheorem{lemma}{Lemma}
\newtheorem{remark}{Remark}
\newtheorem{definition}{Definition}
\newtheorem{assumption}{Assumption}

\newcommand{\real}{\mathbb{R}}
\newcommand{\exreal}{\bar{\mathbb{R}}}

\newcommand{\mathset}[1]{\mathcal{#1}}

\newcommand{\indicator}[1]{\mathcal{I}_{#1}}

\DeclareMathOperator{\epi}{epi}
\DeclareMathOperator{\relint}{relint}
\DeclareMathOperator{\setint}{int}
\DeclareMathOperator{\setbd}{bd}
\DeclareMathOperator{\dom}{dom}
\DeclareMathOperator*{\argmin}{argmin} 
\DeclareMathOperator*{\argmax}{argmax} 

\newcommand{\horizon}{N} 
\newcommand{\Nx}{p} 
\newcommand{\Ny}{\rho}

\newcommand{\Nu}{q} 
\newcommand{\Nw}{\gamma}

\newcommand{\Nlambda}{r} 

\newcommand{\constsc}{m}
\newcommand{\constls}{L}
\newcommand{\constlipschitz}{\bar{L}}

\newcommand{\hidetext}[1]{}

\newcommand{\changes}[1]{{\color{black} #1}}

\def\BibTeX{{\rm B\kern-.05em{\sc i\kern-.025em b}\kern-.08em
    T\kern-.1667em\lower.7ex\hbox{E}\kern-.125emX}}
\begin{document}

\title{
Dynamic Incentive Selection for Hierarchical Convex Model Predictive Control
}

\author{
Akshay Thirugnanam and Koushil Sreenath
\thanks{
This work was supported in part through funding from the Tsinghua-Berkeley Shenzhen Institute (TBSI) program.\\
The authors are with the Department of Mechanical Engineering, UC Berkeley, CA 94720, USA. \tt\small\{akshay\_t, koushils\}@berkeley.edu}
\vspace{-15pt}
}

\maketitle

\begin{abstract}
In this paper, we discuss incentive design for hierarchical model predictive control (MPC) systems viewed as Stackelberg games.
We consider a hierarchical MPC formulation where, given a lower-level convex MPC (LoMPC), the upper-level system solves a bilevel MPC (BiMPC) subject to the constraint that the lower-level system inputs are optimal for the LoMPC.
Such hierarchical problems are challenging due to optimality constraints in the BiMPC formulation.
We analyze an incentive Stackelberg game variation of the problem, where the BiMPC provides additional incentives for the LoMPC cost function, which grants the BiMPC influence over the LoMPC inputs.
We show that for such problems, the BiMPC can be reformulated as a simpler optimization problem, and the optimal incentives can be iteratively computed without knowing the LoMPC cost function.
We extend our formulation for the case of multiple LoMPCs and propose an algorithm that finds bounded suboptimal solutions for the BiMPC.
We demonstrate our algorithm for a dynamic price control example, where an independent system operator (ISO) sets the electricity prices for electric vehicle (EV) charging with the goal of minimizing a social cost and satisfying electricity generation constraints.
Notably, our method scales well to large EV population sizes.
\end{abstract}

\begin{IEEEkeywords}
Model predictive control,
bilevel optimization,
incentive Stackelberg games,
convex analysis,
dynamic pricing,
electric vehicle charging.
\end{IEEEkeywords}

\section{Introduction}
\label{sec:introduction}


Complex control subsystems that interact with each other can be broadly classified into two categories, depending on the type of information flow: distributed and hierarchical.
In distributed systems, information is passed between local regulators, while in hierarchical systems, a high-level coordinating algorithm passes information to local regulators; see ~\cite{scattolini2009architectures} for classifications.
Stackelberg games are leader-follower games where the leader announces its strategy to the followers, with the goal of minimizing the leader cost function~\cite{simaan1973stackelberg}.
We consider a hierarchical MPC problem where each lower-level system solves a convex MPC problem (LoMPC) to determine its control input.
The upper-level system solves a bilevel MPC problem (BiMPC), which selects the upper-level system inputs subject to the constraint that the lower-level system inputs are optimal for the LoMPC.
Such bilevel optimization problems can be viewed as Stackelberg games and are challenging due to the implicit optimality constraints.
In this paper, we analyze an incentive Stackelberg game variation, where the BiMPC provides additional incentives to the LoMPC that allow the BiMPC to steer the LoMPC inputs.

\begin{figure}
    \centering
    \includegraphics[width=0.99\columnwidth]{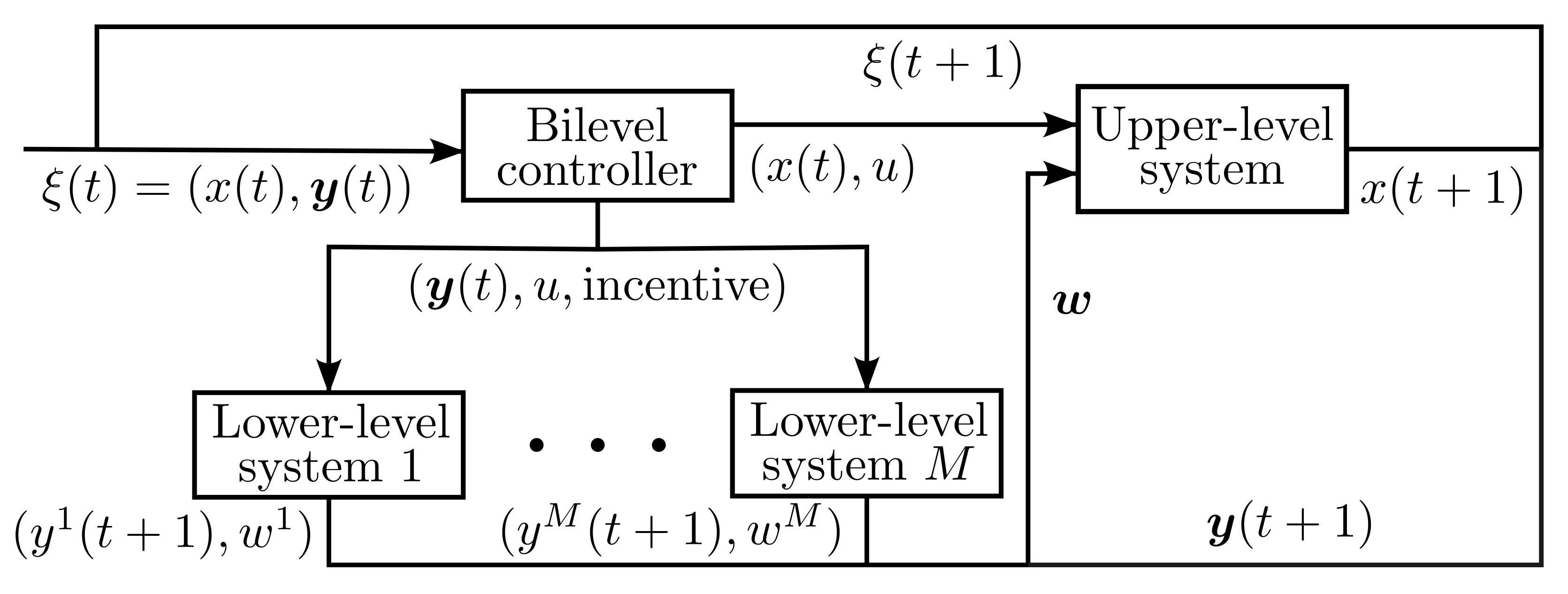}
    \caption{A flowchart depicting the incentive hierarchical MPC formulation.
    Given the current state $\xi(t) = (x(t), \bm{y}(t))$, the bilevel controller (the BiMPC) provides upper-level input $u$ and an incentive to the lower-level systems, which compute the lower-level input $\bm{w}$ using the LoMPCs.
    We assume that the BiMPC can query the lower-level system inputs for different incentives. 
    The controller output $(u, \bm{w})$ is then used to update the upper-level system state.
    The BiMPC and the LoMPCs together comprise the hierarchical MPC problem.
    In the EV charging example considered in Secs.~\ref{subsec:motivating-example} and \ref{sec:numerical-examples}, the BiMPC is solved by an ISO to determine the amount of electricity generation, while the LoMPCs are solved by EVs to determine charging rates.
    The incentive provided by the ISO is the unit price of electricity.
    }
    \label{fig:hierarchical-mpc-formulation}
\end{figure}

\subsection{Related Work}
\label{subsec:related-work}

\subsubsection{Distributed MPC}
\label{subsubsec:distributed-mpc}

In distributed MPC, the information shared between the local regulators consists of predicted system states or control inputs~\cite{scattolini2009architectures,negenborn2014distributed}.
The local regulators can be cooperative~\cite{amigoni2007formal,venkat2008distributed}, or non-cooperative~\cite{betti2013distributed}.
Existing works also consider sparse communication topologies~\cite{negenborn2007efficient}.
Some standard solution methods for distributed MPCs are robust MPC techniques~\cite{mayne2005robust} (for non-cooperative settings) and operator splitting methods like ADMM~\cite{boyd2011distributed} (for cooperative settings).
Control theoretical properties such as stability have also been studied for distributed MPC~\cite{venkat2005stability}.

\subsubsection{Hierarchical MPC}
\label{subsubsec:hierarchical-mpc}

In hierarchical MPC, the high-level coordinator shares its predicted control signal with the local regulators, and the goal is to optimize a high-level cost function~\cite{scattolini2007hierarchical,findeisen1980control}.
Solution techniques for hierarchical MPC include ADMM~\cite{braun2018hierarchical,zhang2020improved}, multiparametric programming~\cite{gupta2023hierarchical}, waypoint tracking~\cite{koeln2018two,koeln2020vertical}, tube-MPC~\cite{kogel2023safe,raghuraman2022hierarchical,picasso2010mpc}, and aggregation methods~\cite{johansson2012distributed}.
Price coordination is another class of methods relevant to our work, which uses Lagrange multipliers to link constraints across various subsystems.
The prices are iteratively computed to ensure global consensus among all the participating subsystems~\cite{mesarovic2000theory, findeisen1980control, negenborn2007efficient, amigoni2007formal, tatara2007control, negenborn2007multi}.

\subsubsection{Incentive Stackelberg Games}
\label{subsubsec:incentive-Stackelberg-games}

Stackelberg games can be formulated as bilevel optimization problems, and optimal leader strategies can be obtained by reformulating the bilevel problem to a single-level problem~\cite{zugno2013bilevel,jamaludin2015bilevel,ouattara2018duality}.
The control theoretical properties of hierarchical MPC as a Stackelberg game are studied in~\cite{mintz2018control}.
A subclass of Stackelberg games is incentive Stacklberg games, where, in addition to the leader strategy, an incentive is provided to the follower by the leader~\cite{mukaidani2020incentive,ho1982control,li2002approach}.
The game is said to be \emph{incentive controllable} if the leader can find incentives such that, given the leader strategy, the optimal follower strategy coincides with the optimal team strategy~\cite{ho1982control} (also see \cite[Sec.~2]{mukaidani2020incentive}).
In this paper, we extend the discussion in \cite{mintz2018control} to incentive Stackelberg games.
We define and discuss sufficient conditions for incentive controllability from the perspective of MPC.

\subsubsection{Dynamic Pricing for EV Charging}
\label{subsubsec:dynamic-pricing-for-ev-charging}

We consider a dynamic price control example for electric vehicle (EV) charging to demonstrate the algorithm proposed in this paper.
Dynamic price control is one approach to reconcile the mismatch between electricity supply and demand at different times of the day and can be formulated as a Stackelberg game~\cite{zugno2013bilevel,mintz2018control,ma2013decentralized,zou2017efficient}.
Dynamic price control aims to reduce the spikes in electricity generation by shifting the electricity demand to different times of the day.
Ideally, the electricity consumption pattern obtained using dynamic pricing achieves demand valley-filling~\cite{ma2013decentralized,zou2017efficient}.

\hidetext{
As a remark, another field of literature that uses the leader-follower connections discussed in this paper for optimal control is Hierarchical Reinforcement Learning (HRL)~\cite{pateria2021hierarchical}.
In HRL, the leader RL controller computes a leader control and a desired state for the follower RL controller.
The method proposed in this paper is similar to this approach in the sense that the incentives provided by the leader MPC define a desired input sequence for the follower MPC.
Existing literature studies the suboptimality of HRL for different timescales of the subsystems~\cite{nachum2018near}, using concepts similar to incentive controllability.
HRL has been applied for the control of hierarchical systems, for example, for quadrotor control~\cite{zhao2021hierarchical, Li2021ASH}.
}


\subsection{Contributions}
\label{subsec:contribution}

The contributions of our paper are as follows:
(a) We extend the discussion in \cite{mintz2018control} on hierarchical MPC as Stackelberg games to incentive Stackelberg games.
For the hierarchical MPC problem, we define incentive controllability and team-optimal solutions.
An incentive that achieves the team-optimal solution is defined as an optimal incentive.
We show that for the case of a single LoMPC, the
optimal linear-convex incentives can be found iteratively without knowing the LoMPC cost function.
Our analysis relies on tools from convex analysis and duality theory.
\hidetext{Our algorithm is similar to the dual ascent algorithm with augmented Lagrangians for convex optimization problems~\cite[Sec.~2.3]{boyd2011distributed}; we make this connection concrete in the later sections.}
(b) For the case of multiple LoMPCs, we compute linear-convex incentives that achieve the team-optimal solution up to an error bound.
We provide a single-level robust MPC reformulation of the BiMPC, whose optimal solution has bounded error compared to the team-optimal solution.
(c) We demonstrate our algorithm for dynamic price control for EV charging.
Our algorithm scales well with the EV population size and does not require the EV cost function.

\subsection{Paper Structure}
\label{subsec:paper-structure}

The paper is structured as follows:
Sec.~\ref{sec:background-and-problem-formulation} provides background on convex analysis and motivates and states the hierarchical MPC problem.
Sec.~\ref{sec:linear-incentives-for-hierarchical-mpc} defines linear incentives and incentive controllability for the hierarchical MPC and provides an iterative algorithm to calculate the optimal incentives.
Sec.~\ref{sec:linear-convex-incentives} extends the analysis in Sec.~\ref{sec:linear-incentives-for-hierarchical-mpc} to linear-convex incentives.
Sec.~\ref{sec:multiple-lompcs} considers the case where the hierarchical MPC problem with linear-convex incentives consists of multiple LoMPCs and provides a robust MPC formulation. 
A dynamic price control example for EV charging is considered in Sec.~\ref{sec:numerical-examples} and conclusions are presented in Sec.~\ref{sec:conclusions}.

\subsection{Notations}
\label{subsec:notations}

The following notations are used throughout the paper:
%
$\exreal = \real \cup \{-\infty, +\infty\}$ is the set of extended real numbers.
The Minkowski sum of two sets $\mathset{A}, \mathset{B} \subset \real^n$ is $\mathset{A} + \mathset{B} = \{a+b: a \in \mathset{A}, b \in \mathset{B}\}$.
$\langle \cdot, \cdot \rangle$ represents an arbitrary inner product on the vector space considered, and $\lVert \cdot \rVert$ is the corresponding norm.
Note that the transpose of a linear map $A: \real^n \rightarrow \real^m$, $A^\top$, depends on the inner products defined on the domain $\real^n$ and co-domain $\real^m$.
$\dom{f}$ is the domain of the function $f$. 
For a set $\mathset{C}$, $\relint{\mathset{C}}$ is the relative interior of $\mathset{C}$ \cite[Sec.~2.1.3]{boyd2004convex}.
$\indicator{\mathset{C}}: \real^n \rightarrow \exreal$ is the indicator function of a set $\mathset{C} \subset \real^n$, where $\indicator{\mathset{C}}(x) = 0 \ \forall x \in \mathset{C}$ and $\indicator{\mathset{C}}(x) = \infty \ \forall x \notin \mathset{C}$.
%
We also define $\langle \horizon \rangle := \{0, ..., \horizon-1\}$ and $[\horizon] := \{0, ..., \horizon\}$.
For a function $\phi: \real^n \rightarrow \real^r$ and $x \in \dom{\phi}$, $D\phi(x): \real^n \rightarrow\real^r$ is the derivative of $\phi$.
For a real-valued function $f: \real^n \rightarrow \real$, $\nabla f$ is the gradient of $f$ (when $f$ is differentiable) and $\partial f$ is the subdifferential of $f$ (see Sec.~\ref{subsubsec:subdifferential-of-a-convex-function}).
Subscripts such as $D_x \phi, \partial_x f$ denote the differentiation variable.


\section{Background and Problem Formulation}
\label{sec:background-and-problem-formulation}

%




\subsection{Convex Analysis}
\label{subsec:convex-analysis}

For a function $f: \real^n \rightarrow \exreal$, the epigraph of $f$ is defined as $\epi{f} = \{(x, t) \in \real^{n+1}: f(x) \leq t, x \in \dom{f}\}$.
Note that $\epi{f} \subset \real^{n+1}$.
A function $f: \real^n \rightarrow \real \cup \{\infty\}$ is a closed proper convex function if $\epi{f}$ is a closed, nonempty, and convex set respectively~\cite[Sec.~4, Sec.~7]{rockafellar1997convex}.

\subsubsection{Subdifferential of a Convex Function}
\label{subsubsec:subdifferential-of-a-convex-function}

The subdifferential of a convex function $f: \real^n \rightarrow \exreal$ at $x$ is~\cite[Sec.~23]{rockafellar1997convex}
\begin{equation}
\label{eq:subdifferential-definition}
    \partial{f}(x) = \{v \in \real^n: f(z) \geq f(x) + \langle v, z-x\rangle, \ \forall z\}.
\end{equation}
$\partial{f}$ is called the subdifferential of $f$.
For all $x$, $\partial{f}(x)$ is a closed, convex set~\cite[Sec.~23]{rockafellar1997convex}.
A proper convex function $f$ is subdifferentiable on $\relint(\dom{f})$, and $\partial{f}(x)$ is non-empty and bounded if and only if $x \in \setint(\dom{f})$~\cite[Th.~23.4]{rockafellar1997convex}.

\subsubsection{Conjugate of a Convex Function}
\label{subsubsec:conjugate-of-a-convex-function}

For a convex function $f$, its conjugate $f^*: \real^n \rightarrow \exreal$ is defined as~\cite[Sec.~12]{rockafellar1997convex}
\begin{equation}
\label{eq:conv-analysis-conjugate}
    f^*(v) = \textstyle{\sup}_x \ (\langle x, v\rangle - f(x)).
\end{equation}
For a proper convex function $f$, $f^*$ is a closed proper convex function~\cite[Thm.~12.2]{rockafellar1997convex}.
For a closed proper convex function $f$, $f^{**} = f$~\cite[Cor.~12.2.1]{rockafellar1997convex}, \changes{and the subdifferential of $f^*$ is given by~\cite[Thm.~23.5]{rockafellar1997convex}}
\begin{equation}
\label{eq:conv-analysis-conjugate-subgrad}
    \partial{f^*}(v) = \textstyle{\argmax}_x \ (\langle x, v\rangle - f(x)).
\end{equation}

\subsubsection{Strong Convexity and Lipschitz Smoothness}
\label{subsubsec:strong-convexity-and-lipschitz-smoothness}

A proper function $f: \real^n \rightarrow \real \cup \{\infty\}$ is strongly convex with modulus $\constsc > 0$ \changes{(abbreviated as $\constsc$-strongly convex)}, if $\dom{f}$ is convex and for all $x \in \dom{f}$ and $v \in \partial f(x)$,
\begin{equation}
    f(z) \geq f(x) + \langle v, z - x\rangle + (\constsc/2) \lVert z - x\rVert^2, \quad \forall z.
\end{equation}
\changes{A function $f$ is $\constsc$-strongly convex if and only if $f - (\constsc/2)\lVert \cdot\rVert^2$ is convex~\cite[Prop.~IV.1.1.2]{hiriart1996convex}.}
Strongly convex functions with closed domains have unique minima.
So by \eqref{eq:conv-analysis-conjugate-subgrad}, the conjugate of a strongly convex function is differentiable (see Prop.~\ref{prop:duality} below).
A differentiable function $f: \real^n \rightarrow \exreal$ is Lipschitz smooth with constant $\constls > 0$ \changes{(abbreviated as $\constls$-smooth)} if $\nabla f$ is Lipschitz continuous with constant $\constls$ on $\dom{f}$, i.e.,
\begin{equation}
    \lVert \nabla f(x^1) - \nabla f(x^2) \rVert \leq \constls \lVert x^1 - x^2 \rVert, \ \forall x^1, x^2 \in \dom{f}.
\end{equation}

There is a duality between strongly convex functions and Lipschitz smooth convex functions as follows. 
\begin{proposition}~\cite[Thm.~X.4.2.1, Thm.~X.4.2.2]{hiriart1993convex}
\label{prop:duality}
    If a proper function $f: \real^n \rightarrow \real \cup \{\infty\}$ is $\constsc$-strongly convex, then $\dom{f^*} = \real^n$ and $f^*$ is $(1/\constsc)$-smooth.
\end{proposition}

\subsubsection{Dual Cones}
\label{subsubsec:dual-cones}

Let $\mathset{K} \subset \real^n$ be a closed convex cone.
Then, $\mathset{K}$ defines a relation $\preceq_\mathset{K}$ on $\real^n$ as follows:
$x^1 \preceq_\mathset{K} x^2$ if and only if $x^2-x^1 \in \mathset{K}$~\cite[Sec.~2.4]{boyd2004convex}.
\hidetext{Note that $\preceq_\mathset{K}$ is not a partial order on $\real^n$ (it is transitive and reflexive but might not be antisymmetric).}
The dual cone $\mathset{K}^*$ of a cone $\mathset{K}$ is defined as~\cite[Sec.~2.6]{boyd2004convex}
\begin{equation}
\label{eq:dual-cone-def}
    \mathset{K}^* = \{v: \langle v, x\rangle \geq 0, \forall x \in \mathset{K}\}.
\end{equation}
The dual cone $\mathset{K}^*$ depends on the inner product chosen on $\real^n$.
For any cone $\mathset{K}$, $\mathset{K}^*$ is a closed convex cone. 

\subsubsection{Convex Vector Functions}
\label{subsubsec:convex-vector-functions}

Let $\mathset{K} \subset \real^m$ be a closed convex cone.
A function $f: \real^n \rightarrow \real^m$ is called $\mathset{K}$-convex if $\dom{f}$ is convex and for all $x^1, x^2 \in \dom{f}$ and $\alpha \in [0, 1]$,
\begin{equation}
\label{eq:k-convex-function}
    f(\alpha x^1 + (1-\alpha) x^2) \preceq_\mathset{K} \alpha f(x^1) + (1-\alpha) f(x^2).
\end{equation}
For any $\mathset{K}$-convex function $f$ and $\lambda \succeq_{\mathset{K}^*} 0$, $\langle \lambda, f(\cdot) \rangle$ is a convex function~\cite[Sec.~3.6.2]{boyd2004convex}.

In the subsequent sections, we will use convex analysis to compute optimal incentives for the hierarchical MPC problem.
Next, we motivate and describe the hierarchical MPC problem.

\subsection{Motivating Example}
\label{subsec:motivating-example}

As a motivating example for the hierarchical MPC formulation in the next subsection, we consider the charging of a population of EVs.
Given the unit price of electricity, each EV determines its electricity consumption by considering the total price of electricity, battery degradation, and current state of charge (SoC).
For the same electricity price, each EV can have different consumption depending on its SoC.
An independent system operator (ISO) determines the amount of electricity generated, the amount stored in a storage battery,
and the electricity price.
The ISO's goal is to minimize a social cost, which comprises the cost of electricity generation and the cost of not meeting the EV consumption demand.
In particular, the ISO's cost only considers the average EV consumption demand (computed using current EV SoCs); for the same average consumption, individual consumption may depend on the EV SoC.
Similarly, the ISO determines electricity generation based only on the average consumption.

Such a hierarchical problem is an incentive Stackelberg game (see Sec.~\ref{subsubsec:incentive-Stackelberg-games}), where the leader (the ISO) provides incentives (electricity prices) to the followers (the EV population) to minimize the social cost.
The following subsection provides a general hierarchical MPC formulation applicable to the above example.
In Secs.~\ref{sec:linear-incentives-for-hierarchical-mpc}-\ref{sec:multiple-lompcs}, we develop an algorithm to solve the hierarchical MPC problem, and in Sec.~\ref{sec:numerical-examples}, we apply our method to the EV charging example.
We will occasionally refer to the EV charging example to provide intuition.

\subsection{Hierarchical MPC Formulation}
\label{subsec:hierarchical-formulation}

We consider a hierarchical control system consisting of multiple lower-level systems and an upper-level system as a generalization of the dynamics in \cite{mintz2018control}.
Let $\mathset{M}$ be an index set denoting the set of lower-level systems, with $M := |\mathset{M}|$.
For $i \in \mathset{M}$, the $i$-th lower-level system has state $y^i \in \real^\Ny$ and 
\changes{input $(u, w^i) \in \real^\Nu \times \real^\Nw$}.
We define $\bm{w} := (w^1, ..., w^{M})$.
The upper-level system has state $x \in \real^\Nx$ and input $\nu := (u, \bm{w})$.
The hierarchical system has coupled linear dynamics given by,
\begin{equation}
\label{eq:coupled-dynamics}
\begin{split}
    x_{k+1} & = A^u x_k + B^u_1 u_k + B^u_2 w_k, \\
    y^i_{k+1} & = A^l y^i_k + B^l_1 u_k + B^l_2 w^i_k, \quad i \in \mathset{M}, \\
    w_k & := \frac{1}{M} \sum_{i \in \mathset{M}} w^i_k,
\end{split}
\end{equation}
\changes{where the superscripts `u' and `l' refer to the upper and lower-level systems, respectively.
Variable subscripts represent time, and superscripts are used for the index set $\mathset{M}$.
The superscript $(\cdot)^*$ denotes an optimal value, while the superscript $(\cdot)^{(k)}$ denotes the $k$-th iterate of a sequence.}
We define $\bm{y} := (y^1, ..., y^{M})$ as the joint state of all the lower-level systems, and $\xi := (x, \bm{y})$ as the state of the entire system.
Note that the upper-level system dynamics depend on the average $w_k$, similar to the EV charging example in Sec.~\ref{subsec:motivating-example}.

\subsubsection{Lower-level MPC Problem}
\label{subsubsec:lower-level-mpc-problem}

The $i$-th lower-level MPC (LoMPC) problem with horizon $\horizon$ optimizes the input sequence $w^i_{0:\horizon-1} = (w^i_0, ..., w^i_{\horizon-1}) \in \real^{\Nw\horizon}$, given the initial state $\xi_0$ and the input sequence $u_{0:\horizon-1} = (u_0, ..., u_{\horizon-1}) \in \real^{\Nu\horizon}$.
The symbols $\xi$, $x$, $\bm{y}$, $\nu$, $u$, and $\bm{w}$ are overloaded to also represent $\xi_{0:\horizon}$, $x_{0:\horizon}$, $\bm{y}_{0:\horizon}$, $\nu_{0:\horizon-1}$, $u_{0:\horizon-1}$, and $\bm{w}_{0:\horizon-1}$ respectively.
The variables used in the hierarchical MPC are tabulated in Tab.~\ref{tab:hierarchical-mpc-variable-definitions}.
The $i$-th LoMPC problem is given by,
\begin{subequations}
\label{eq:ith-lompc}
\begin{alignat}{6}
    \min_{y^i, w^i} \ & g^i(y^i, w^i; \xi_0, u), \span\span\span\span\span\span \label{subeq:ith-lompc-cost} \\
    \text{s.t} \quad &
    && y^i_{k+1}   &&= A^l y^i_k + B^l_1 u_k + B^l_2 w^i_k, \ \ && k \in \langle \horizon \rangle, \label{subeq:ith-lompc-dynamics} \\
    &&& y^i_k \in \mathset{Y}, \ w^i_k \in \mathset{W}, \span\span \ && k \in \langle \horizon \rangle, \label{subeq:ith-lompc-stage-constraints} \\
    &&& y^i_N \in \mathset{Y}_\Omega, \span\span\span\span\span\span \label{subeq:ith-lompc-terminal-constraints}
\end{alignat}
\end{subequations}
where $\mathset{Y}, \ \mathset{Y}_\Omega\subset \real^\Ny$ and $\mathset{W} \subset \real^\Nw$ are closed convex sets, and $\langle \horizon \rangle := \{0, ..., \horizon-1\}$.
The function $g^i$ is parametrized by the input $u$ from the bilevel MPC and $\xi_0$.
The constraints in \eqref{eq:ith-lompc} can be added to the cost using the indicator function (see Sec.~\ref{subsec:notations}) to obtain a reduced LoMPC problem as follows:
\begin{equation}
\label{eq:reduced-ith-lompc}
w^{i*}(\xi_0, u) := \textstyle{\argmin}_{w^i} \ g^i(w^i; \xi_0, u),
\end{equation}
where the symbol $g^i$ is overloaded to represent the new cost function, where the constraints from \eqref{eq:ith-lompc} are added to the cost \changes{using the indicator function}.
We will work with the simple formulation of the cost function, \eqref{eq:reduced-ith-lompc}, for the theory and use the full formulation, \eqref{eq:ith-lompc}, for the implementation.
The assumptions on $g^i$ are as follows:
\begin{assumption} (Properties of $g^i$)
\label{assum:properties-of-gi}
    $g^i(\cdot \, ; \xi_0, u)$ is a closed proper function and $\constsc$-strongly convex, for all $(\xi_0, u)$.
    Further, $g^i(\cdot \: ; \xi_0, u)$ is subdifferentiable on its domain for each $(\xi_0, u)$.
\end{assumption}
Since $g^i(\cdot \,; \xi_0, u)$ is closed and strongly convex, a unique optimum for $w^i$ exists, given \eqref{eq:ith-lompc} is feasible (see Sec.~\ref{subsubsec:strong-convexity-and-lipschitz-smoothness}).

For the EV charging example in Secs.~\ref{subsec:motivating-example} and \ref{sec:numerical-examples}, $y^i_k$ is the SoC and $w^i_k$ is the charging rate of EV $i$ at time $k$.
The LoMPC cost function $g^i$ comprises tracking costs for maximum SoC and costs for high charging rates.

\subsubsection{Bilevel MPC Problem}
\label{subsubsec:bilevel-mpc-problem}

Given the initial state $\xi_0$, the bilevel MPC (BiMPC) problem optimizes the input sequence $\nu = (u, \bm{w})$, subject to the constraint that $w^i$ is obtained as the optimal solution of the LoMPC problem \eqref{eq:reduced-ith-lompc}.
In other words, the BiMPC accounts for the fact that each LoMPC solves an optimization problem to determine its input.
The BiMPC problem is given by,
\begin{subequations}
\label{eq:bimpc}
\begin{alignat}{6}
    \min_{x, \nu=(u,\bm{w})} \ & f(x, u, w; \xi_0), \span\span\span\span\span\span \label{subeq:bimpc-cost} \\
    \text{s.t} \quad &
    && w^i = \textstyle{\argmin}_{w^i} \ g^i(w^i; \xi_0, u), && i \in \mathset{M}, \label{subeq:bimpc-w-optimality} \\
    &&& x_{k+1} = A^u x_k + B^u_1 u_k + B^u_2 w_k, \ \ && k \in \langle \horizon \rangle, \label{subeq:bimpc-dynamics} \\ 
    &&& w_k = \frac{1}{M} \sum_{i \in \mathset{M}} w^i_k, \ \ && k \in \langle \horizon \rangle, \span\span\span\span \\
    &&& x_k \in \mathset{X}, \ u_k \in \mathset{U}, \ && k \in \langle \horizon \rangle, \label{subeq:bimpc-stage-constraints} \\
    &&& x_N \in \mathset{X}_\Omega, \span\span\span\span\span\span \label{subeq:bimpc-terminal-constraints}
\end{alignat}
\end{subequations}
where $\mathset{X}, \ \mathset{X}_\Omega \subset \real^\Nx$ and $\mathset{U} \subset \real^\Nu$ are closed convex sets, and $f$ is jointly continuous in $(x, u, w)$.
Note that the cost function $f$ depends on the average $w = (w_1, ..., w_{N-1})$.
Similar to the LoMPC problem, the constraints \eqref{subeq:bimpc-dynamics}, \eqref{subeq:bimpc-stage-constraints}, and \eqref{subeq:bimpc-terminal-constraints} can be added to the cost \eqref{subeq:bimpc-cost} using the indicator function to obtain the reduced BiMPC problem as follows:
\begin{subequations}
\label{eq:reduced-bimpc}
\begin{alignat}{6}
    \hspace{-5pt} \nu^*(\xi_0) \in \ && \argmin_{\nu=(u, \bm{w})} \ & f(u, w; \xi_0), \span\span\span\span \label{subeq:reduced-bimpc-cost} \\
    && \ \text{s.t} \quad &
    w^i &&= \textstyle{\argmin}_{w^i} \, g^i(w^i; \xi_0, u), \, && i \in \mathset{M}, \hspace{-5pt} \label{subeq:reduced-bimpc-w-optimality} \\
    &&& w_k &&= \frac{1}{M} \sum_{i \in \mathset{M}} w^i_k, && \hspace{-5pt} k \in \langle N\rangle. \hspace{-5pt}
\end{alignat}
\end{subequations}
The optimization problems \eqref{eq:reduced-ith-lompc} and \eqref{eq:reduced-bimpc} define the hierarchical MPC problem.
At each time step, the current state $\xi(t)$ is provided to the BiMPC \eqref{eq:reduced-bimpc}, and the first input $\nu^*_0$ from the optimal input sequence $\nu^*$ is applied to get the next state.
Throughout the paper, we assume that \eqref{eq:reduced-bimpc} is feasible and has an optimal solution for all $\xi_0$.

For the EV charging example in Secs.~\ref{subsec:motivating-example} and \ref{sec:numerical-examples}, $x_k$ is the SoC of the storage battery controlled by the ISO at time $k$ and $u_k$ is the amount of electricity generated.
The dynamics \eqref{subeq:bimpc-dynamics} corresponds to energy balance, and the cost function $f$ is a social cost comprising electricity generation cost and EV electricity demand cost.

Dividing the problem of control selection for the hierarchical system into a multi-level problem such as \eqref{eq:reduced-bimpc} can result in a loss of controllability or stabilizability properties of the overall hierarchical system if the cost functions $g^i$ are selected inappropriately~\cite{mintz2018control}.

\begin{table}[tbp]
\setlength\extrarowheight{2pt}
\caption{Hierarchical MPC variable definitions}\label{tab:hierarchical-mpc-variable-definitions}
\begin{tabularx}{\columnwidth}{|c|c|L|}
\hline
Variable & Domain & Definition \\ \hline
$(y^i_k, w^i_k, u_k)$ & $\real^{\Ny+\Nw+\Nu}$ & State and input of $i$-th lower-level system at time $k$ \\
$(y^i, w^i)$ & $\real^{\Ny(\horizon+1)+\Nw\horizon}$ & $\bigl((y^i_0, ..., y^i_\horizon), (w^i_0, ..., w^i_{\horizon-1})\bigr)$ \\
$(\bm{y}, \bm{w})$ & $\real^{\Ny(\horizon+1) M+\Nw \horizon M}$ & $\bigl((y^1, ..., y^M), (w^1, ..., w^M)\bigr)$ \\
$(w, w_k)$ & $\real^{\Nw\horizon+\Nw}$ & Averages: $\bigl(1/M \, \Sigma_i w^i, 1/M \, \Sigma_i w^i_k\bigr)$ \\ \hline
$(x_k, u_k, \bm{w}_k)$ & $\real^{\Nx+\Nu+\Nw M}$ & State and input of upper-level system at time $k$ \\
$(x, u)$ & $\real^{\Nx(\horizon+1)+\Nu\horizon}$ & $\bigl((x_0, ..., x_\horizon), (u_0, ..., u_{\horizon-1})\bigr)$ \\ \hline
$\xi$ & $\real^{(\Nx+\Ny M)(\horizon+1)}$ & $(x, \bm{y})$, state of the full system \\
$\nu$ & $\real^{(\Nu+\Nw M)\horizon}$ & $(u, \bm{w})$, input of the full system \\
\hline
\end{tabularx}
\end{table}

The hierarchical MPC problem given by \eqref{eq:reduced-ith-lompc} and \eqref{eq:reduced-bimpc} is a Stackelberg game, where the BiMPC \eqref{eq:reduced-bimpc} is the leader, and the LoMPCs \eqref{eq:reduced-ith-lompc} are the followers.
As discussed in Sec.~\ref{subsubsec:incentive-Stackelberg-games}, we consider the incentive Stackelberg game variation, where the BiMPC provides an additional incentive to the LoMPCs.
Incentive variables impose penalties (or rewards) on the LoMPCs in addition to the cost function $g^i$ and thus grant the BiMPC influence over the LoMPC output.
Fig.~\ref{fig:hierarchical-mpc-formulation} depicts a flowchart of the incentive hierarchical MPC problem.
We make the following assumption about the incentive hierarchical MPC problem, formally defined in Sec.~\ref{subsec:linear-incentives-and-incentive-controllability}.

\begin{assumption} (Privacy of LoMPCs)
\label{assum:privacy-of-lompc}
    The cost functions $g^i$ are unknown to the BiMPC, apart from $\dom{g^i}$.
    The BiMPC can query the optimal solution $w^{i*}$ from the $i$-th LoMPC for any given $(\xi_0, u)$ and incentive.
\end{assumption}

In the following two sections, we first consider the case of a single LoMPC (i.e., when $M = 1$).
In Sec.~\ref{sec:linear-incentives-for-hierarchical-mpc}, we define a linear incentive structure and the notion of incentive controllability.
We prove that an optimal incentive for the incentive hierarchical MPC can be computed iteratively, even when the LoMPC cost function is unknown.
In Sec.~\ref{sec:linear-convex-incentives}, we extend the formulation in Sec.~\ref{sec:linear-incentives-for-hierarchical-mpc} to linear-convex incentives.
In Sec.~\ref{sec:multiple-lompcs}, we consider the case of multiple LoMPCs (i.e., when $M > 1$) with linear-convex incentives and propose a robust formulation for the incentive hierarchical MPC.

\section{Linear Incentives for Hierarchical MPC}
\label{sec:linear-incentives-for-hierarchical-mpc}

In this section, we propose an incentive hierarchical MPC that adds a linear incentive term to the LoMPC cost function.
For the EV charging example, the total electricity price is a linear incentive for the EVs, as explained in the following subsection.
We consider the case of a single LoMPC (i.e., when $M = 1$) and denote the state, input, and cost function of the LoMPC by $y$, $w$, and $g$, respectively.
We define the notion of incentive controllability and show that the incentive hierarchical MPC satisfies this property.
We also provide an iterative method to compute an optimal incentive without knowing the LoMPC cost function $g$.

\subsection{Linear Incentives and Incentive Controllability}
\label{subsec:linear-incentives-and-incentive-controllability}

Let $g(w; \xi_0, u, \lambda)$ be the incentive LoMPC cost function given by,
\begin{equation}
\label{eq:incentivized-lompc-cost}
    g(w; \xi_0, u, \lambda) := g(w; \xi_0, u) + \langle \lambda, w\rangle,
\end{equation}
where $\lambda \in \real^{\Nw\horizon}$ is the incentive to the LoMPC \changes{and $\langle \lambda, w\rangle$ is the linear incentive term.}
For the EV charging example in Secs.~\ref{subsec:motivating-example} and \ref{sec:numerical-examples}, $\lambda$ is the unit price of electricity, and $\langle \lambda, w\rangle$ is the total price paid by the EV.
The LoMPC problem \eqref{eq:reduced-ith-lompc} with an incentive is as follows:
\begin{equation}
\label{eq:incentivized-lompc}
w^*(\xi_0, u, \lambda) := \textstyle{\argmin}_{w} \ (g(w; \xi_0, u) + \langle \lambda, w\rangle),
\end{equation}
where the BiMPC problem provides both $u$ and $\lambda$ as inputs to the LoMPC.
As discussed in Sec.~\ref{subsubsec:bilevel-mpc-problem}, incentives impose additional penalties or rewards on the LoMPC and can be used to influence $w^*$.
By Assumption~\ref{assum:properties-of-gi}, the cost function of the LoMPC problem \eqref{eq:incentivized-lompc} is a closed proper function, $\constsc$-strongly convex, and subdifferentiable for all values of $\lambda$.
So, there is a unique solution to \eqref{eq:incentivized-lompc} for all values of $\lambda$.
The BiMPC problem \eqref{eq:reduced-bimpc} with an incentive is as follows:
\begin{subequations}
\label{eq:incentivized-bimpc}
\begin{alignat}{6}
    (\nu^*, \lambda^*)(\xi_0) \in \ && \argmin_{\nu=(u,w), \lambda} & f(\nu; \xi_0), \label{subeq:incentivized-bimpc-cost} \\
    && \ \text{s.t} \quad &
    w = \textstyle{\argmin}_w \ g(w; \xi_0, u, \lambda). && \label{subeq:incentivized-bimpc-w-optimality}
\end{alignat}
\end{subequations}
The BiMPC dynamically chooses the incentive $\lambda$ at each time step.
The optimization problems \eqref{eq:incentivized-lompc} and \eqref{eq:incentivized-bimpc} define the incentive hierarchical MPC problem.
The incentive BiMPC problem \eqref{eq:incentivized-bimpc} is nontrivial to solve due to the optimality constraint \eqref{subeq:incentivized-bimpc-w-optimality}.
\changes{
In the rest of the section, we discuss a solution method for \eqref{eq:incentivized-bimpc}.
}
%
%
Following \cite[Sec.~7.4]{basar1998dynamic} and \cite[Def.~1]{mukaidani2020incentive}, we first define a team-optimal solution for the incentive hierarchical MPC problem as follows.
\begin{definition} (Team-optimal solution)
\label{def:team-optimal-solution}
For a given $\xi_0$, a team-optimal solution of the incentive hierarchical MPC is an optimal solution $\nu^* = (u^*, w^*)$ of the following problem:
\begin{subequations}
\label{eq:equivalent-bimpc}
\begin{alignat}{6}
    \nu^*(\xi_0) \in \ && \argmin_{\nu=(u,w)} \ & f(\nu; \xi_0), \label{subeq:equivalent-bimpc-cost} \\
    && \ \text{s.t} \quad &
    w \in \dom g(\cdot \; ; \xi_0, u). \label{subeq:equivalent-bimpc-w-feasibility}
\end{alignat}
\end{subequations}
\end{definition}
Recall from Sec.~\ref{subsubsec:lower-level-mpc-problem} that the constraint \eqref{subeq:equivalent-bimpc-w-feasibility} is equivalent to enforcing the state and input constraints \eqref{subeq:ith-lompc-dynamics}-\eqref{subeq:ith-lompc-terminal-constraints}. 
We assume that a team-optimal solution exists for all initial states $\xi_0$.
Next, we define incentive controllability (this definition is similar to the one in \cite[Sec.~7.4]{basar1998dynamic} and \cite{ho1982control}, adapted to the problem considered in this paper).
\begin{definition} (Incentive controllability)
\label{def:incentive-controllability}
The incentive hierarchical MPC is called incentive controllable if, for all $\xi_0$ and corresponding team-optimal solutions $\nu^*=(u^*, w^*)$, there exists an incentive $\lambda^*$ such that $w^*$ is optimal for \eqref{eq:incentivized-lompc} with $(u, \lambda) = (u^*, \lambda^*)$, i.e., $w^* = w^*(\xi_0, u^*, \lambda^*)$.
\end{definition}
The following result states that the incentive hierarchical MPC is incentive controllable and shows that \eqref{eq:incentivized-bimpc} can be solved using a two-step approach.

\begin{theorem} (Linear incentive controllability)
\label{thm:bimpc-equivalence}
Under Assumption~\ref{assum:properties-of-gi}, the incentive hierarchical MPC problem, \eqref{eq:incentivized-lompc} and \eqref{eq:incentivized-bimpc}, is incentive controllable.
Further, the incentive BiMPC \eqref{eq:incentivized-bimpc} is equivalent to \eqref{eq:equivalent-bimpc}.
For any team-optimal solution $\nu^* = (u^*, w^*)$ and $-\lambda^* \in \partial_w g(w^*; \xi_0, u^*)$, $(\nu^*, \lambda^*)$ is an optimal solution of \eqref{eq:incentivized-bimpc}.
\end{theorem}

\begin{proof}
Under Assumption~\ref{assum:properties-of-gi}, $g(\cdot \,; \xi_0, u, \lambda)$ is a closed proper function, $\constsc$-strongly convex, and subdifferentiable.
Let $\nu^*=(u^*,w^*)$ be a team-optimal solution for some $\xi_0$.
Then, $w^*$ (uniquely) minimizes $g(w; \xi_0, u^*, \lambda)$ if and only if $-\lambda \in \partial_w g(w^*; \xi_0, u^*) \neq \emptyset$~\cite[Thm.~23.5]{rockafellar1997convex}.
This means that if $-\lambda^* \in \partial_w g(w^*; \xi_0, u^*)$ then $w^*$ is optimal for \eqref{eq:incentivized-lompc} with $(u, \lambda) = (u^*, \lambda^*)$.
Thus, the incentive hierarchical MPC is incentive controllable.

Next, we show the equivalence of \eqref{eq:incentivized-bimpc} and \eqref{eq:equivalent-bimpc}.
For any feasible solution $(\nu, \lambda)$ of \eqref{eq:incentivized-bimpc}, $\nu$ is also feasible for \eqref{eq:equivalent-bimpc} and has the same cost.
Similarly, if $\nu=(u, w)$ is feasible for \eqref{eq:equivalent-bimpc} and $-\lambda \in \partial_w g(w; \xi_0, u)$, then $(\nu, \lambda)$ is feasible for \eqref{eq:incentivized-bimpc} and has the same cost.
Thus, \eqref{eq:incentivized-bimpc} and \eqref{eq:equivalent-bimpc} are equivalent.
\end{proof}

Thm.~\ref{thm:bimpc-equivalence} shows that \eqref{eq:incentivized-bimpc} can be solved by first finding a team-optimal solution $\nu^*=(u^*, w^*)$ using \eqref{eq:equivalent-bimpc} and then choosing $\lambda^* \in -\partial_w g(w^*; \xi_0, u^*)$ as an optimal incentive for \eqref{eq:incentivized-bimpc}.
However, computing such an optimal incentive $\lambda^*$ for a given $w^*$ requires the explicit knowledge of $g$.
In the following subsection, we show that optimal incentives can be computed iteratively without knowing $g$ or its subdifferential.

\subsection{Iterative Method for Optimal Incentives}
\label{subsec:iterative-method}

To simplify notation, we assume that the team-optimal solution $\nu^*$ is available after solving the optimization problem \eqref{eq:equivalent-bimpc} and we suppress the variables $(\xi_0, u^*)$ in the cost function $g$, and write it as $g(w; \lambda) = g(w) + \langle \lambda, w\rangle$.
Given $w^*$, the goal is to find $\lambda^*$ such that $w^* = \argmin_w g(w; \lambda^*)$.
Such an optimal incentive $\lambda^*$ exists from Thm.~\ref{thm:bimpc-equivalence}, which shows that $\partial g$ can be used to find $\lambda^*$.
However, $\partial g$ is unknown to the BiMPC (see Assumption~\ref{assum:privacy-of-lompc}).
Instead, we can query the optimal solution of the LoMPC \eqref{eq:incentivized-lompc} for any incentive $\lambda$, which is computed as $w^*(\lambda) = \argmin_w g(w; \lambda)$.
The following result shows that access to the LoMPC as a black-box solver is sufficient to iteratively compute incentives $\lambda^{(k)}$, such that $\argmin_w g(w; \lambda^{(k)}) \rightarrow w^*$ as $k \rightarrow \infty$.

\begin{theorem} (Iterative method for linear incentives)
\label{thm:iterative-method}
Let Assumption~\ref{assum:properties-of-gi} hold, and $w^* \in \dom{g}$.
Let $\lambda^{(k)}$, $w^{(k)}$ be iterates satisfying
\begin{subequations}
\label{eq:iterative-method-const-step}
\begin{align}
    w^{(k)} & = w^*(\lambda^{(k)}) = \textstyle{\argmin}_w \ (g(w) + \langle \lambda^{(k)}, w\rangle), \label{subeq:iterative-method-const-step-gradient} \\
    \lambda^{(k+1)} & = \lambda^{(k)} + \constsc (w^{(k)} - w^*), \label{subeq:iterative-method-const-step-flow}
\end{align}
\end{subequations}
where $\constsc$ is the modulus of strong convexity of $g$.
Then, $w^{(k)} \rightarrow w^*$ 
with the rate given by $\lVert w^{(k)} - w^*\rVert = O(1/\sqrt{\constsc k})$.
\end{theorem}

\begin{proof}
\changes{We define a dual function $\bar{g}^*$ corresponding to the closed proper and $\constsc$-strongly convex function $g$ as follows:
\begin{equation}
\label{eq:iterative-method-proof-dual-function}
    \bar{g}^*(\lambda) = \textstyle{\min}_w \ (g(w) + \langle \lambda, w\rangle).
\end{equation}
Note that $\bar{g}^*(\lambda) = -g^*(-\lambda)$, where $g^*$ is the conjugate of $g$ (see Sec.~\ref{subsubsec:conjugate-of-a-convex-function}).
So, by Prop.~\ref{prop:duality}, $\bar{g}^*$ is a $(1/\constsc)$-smooth concave function and $\dom(\bar{g}^*) = \real^{\Nw\horizon}$.
Now consider the function $\tilde{g}^*(\lambda) := \bar{g}^*(\lambda) - \langle w^*, \lambda \rangle$.
$\tilde{g}^*$ is concave and $(1/\constsc)$-smooth.
The gradient of $\tilde{g}^*$ is given by \eqref{eq:conv-analysis-conjugate-subgrad} as
\begin{equation*}
    \nabla \tilde{g}^*(\lambda) = \nabla \bar{g}^*(\lambda) - w^* = \textstyle{\argmin}_w \ (g(w) + \langle \lambda, w\rangle) - w^*.
\end{equation*}
So, if $\lambda$ is a critical point of $\tilde{g}^*$, then it must satisfy
\begin{equation*}
    \nabla \tilde{g}^*(\lambda) = 0 \Rightarrow w^* = \textstyle{\argmin}_w \ (g(w) + \langle \lambda, w\rangle).
\end{equation*}
In other words, $w^* = \argmin_w \ g(w; \lambda)$,
meaning $\lambda$ is an optimal incentive to achieve the solution $w^*$.
At least one critical point of $\tilde{g}^*$ exists because $g$ is assumed to be subdifferentiable (Assumption~\ref{assum:properties-of-gi}) and $w^* \in \dom{g}$ (it is obtained from \eqref{eq:equivalent-bimpc}).

Gradient ascent on $\tilde{g}^*$ with the constant step size of $\constsc$ converges because $\tilde{g}^*$ is $(1/\constsc)$-smooth~\cite[Sec.~1.2.3]{nesterov2018lectures}.
Gradient ascent for $\tilde{g}^*$ is of the form:
\begin{align*}
    \lambda^{(k+1)} = \lambda^{(k)} + m\nabla \tilde{g}^*(\lambda^{(k)}) = \lambda^{(k)} + m(w^*(\lambda^{(k)}) - w^*).
\end{align*}
Then, replacing $w^*(\lambda^{(k)})$ with $w^{(k)}$ in the above equation using \eqref{subeq:iterative-method-const-step-gradient}, we obtain \eqref{eq:iterative-method-const-step}.
Since $(w^{(k)}-w^*)$ is the gradient of $\tilde{g}^*$, it converges to zero with 
the rate given by $\lVert w^{(k)} - w^*\rVert = O(1/\sqrt{\constsc k})$~\cite[Eq.~1.2.22]{nesterov2018lectures}.}
\unskip
\end{proof}

\begin{table*}[tbp]
\footnotesize
\setlength\extrarowheight{2pt}
\caption{Properties of the dual functions and main results}\label{tab:properties-of-the-dual-functions-and-main-results}
\begin{tabularx}{\textwidth}{|l||C|C|C|}
\hline
\multicolumn{1}{|c||}{\multirow{2}{*}{Property}} & \multicolumn{2}{c|}{Single LoMPC ($M = 1$)} & \multicolumn{1}{c|}{Multiple LoMPCs ($M > 1$)} \\
\cline{2-4}
& \multicolumn{1}{c|}{Linear incentives (Sec.~\ref{sec:linear-incentives-for-hierarchical-mpc})} & \multicolumn{1}{c|}{Linear-convex incentives (Sec.~\ref{sec:linear-convex-incentives})} & \multicolumn{1}{c|}{Linear-convex incentives (Sec.~\ref{sec:multiple-lompcs})} \\
\hline
Gradient of the dual function $\bar{g}^*$ & Eq.~\eqref{eq:conv-analysis-conjugate-subgrad} & Lem.~\ref{lem:linear-convex-incentive-conjugate-grad} & $-$ \\
Lipschitz smoothness of $\bar{g}^*$ & Prop.~\ref{prop:duality} & Lem.~\ref{lem:linear-convex-lipschitz-smooth-dual} & $-$ \\
\hline
Incentive controllability (Def.~\ref{def:incentive-controllability}) & Thm.~\ref{thm:bimpc-equivalence} & Assumption~\ref{assum:linear-convex-incentive-controllability} & Lem.~\ref{lem:bounded-incentive-controllability} \\
Iterative method for optimal incentives & Thm.~\ref{thm:iterative-method} & Thm.~\ref{thm:linear-convex-incentives-iterative-method} & Thm.~\ref{thm:linear-convex-incentives-iterative-method} \\
Incentive BiMPC reformulation & Def.~\ref{def:team-optimal-solution} & Def.~\ref{def:team-optimal-solution} & Thm.~\ref{thm:robust-bimpc-formulation} \\
\hline
\end{tabularx}
\end{table*}

To summarize the results in this section, we have shown that adding a linear incentive to the LoMPC cost results in incentive controllability and allows the BiMPC to select an optimal incentive to achieve the team-optimal solution.
Moreover, optimal incentives can be calculated iteratively without knowing the LoMPC cost.
In the next section, we build upon the definitions and results in this section and consider a linear-convex incentive structure, which is a generalization of the linear incentive. 
We show that an analog of the iterative method, Thm.~\ref{thm:iterative-method}, can be used for linear-convex incentives.

\section{Linear-Convex Incentives}
\label{sec:linear-convex-incentives}

In this section, we consider linear-convex incentives of the form $\langle \lambda, \phi(w)\rangle$, which is linear in $\lambda$ and convex in $w$.
We consider the case of a single LoMPC (i.e., when $M = 1$) and assume that $\phi$ is known to the BiMPC.
The results in this section follow the general outline for the linear incentive case in Sec.~\ref{sec:linear-incentives-for-hierarchical-mpc}.
First, we state assumptions on the function $\phi$ and then define a dual function $\bar{g}^*$ in $\lambda$, similar to the conjugate function $g^*$.
We then prove a duality result similar to Prop.~\ref{prop:duality}.
Finally, assuming incentive controllability, we derive an iterative method for linear-convex incentives.
The results in this section are a generalization of the results in Sec.~\ref{sec:linear-incentives-for-hierarchical-mpc}, see Tab.~\ref{tab:properties-of-the-dual-functions-and-main-results} for a comparison.

A linear-convex incentive $\langle \lambda, \phi(w) \rangle$ is a generalization of the linear incentive $\langle \lambda, w\rangle$.
For the EV charging example, linear-convex incentives can be used for nonlinear electricity pricing.
Linear-convex incentives can also be used to encode the fact that unit electricity prices are nonnegative, as we will show in Sec.~\ref{sec:numerical-examples}.
The LoMPC problem \eqref{eq:reduced-ith-lompc} with a linear-convex incentive structure is as follows:
\begin{equation}
\label{eq:linear-convex-incentivized-lompc}
w^*(\xi_0, u, \lambda) := \textstyle{\argmin}_{w} \ g(w; \xi_0, u) + \langle \lambda, \phi(w) \rangle.
\end{equation}

Similar to Sec.~\ref{subsec:iterative-method}, we assume that the team-optimal solution $\nu^*$ is available after solving the optimization problem \eqref{eq:equivalent-bimpc} and drop the dependence of $g$ on $(\xi_0, u^*)$ for the rest of this section.
We define a linear-convex incentive as follows.
\begin{definition}(Linear-convex incentive)
\label{def:linear-convex-incentive}
    Let $\mathset{K} \subset \real^\Nlambda$ be a closed convex cone and $\phi: \real^{\Nw\horizon} \rightarrow \real^\Nlambda$ be a continuously differentiable $\mathset{K}$-convex function (see Sec.~\ref{subsubsec:convex-vector-functions}).
    Then the function $\bar{\phi}: \mathset{K}^* \times \real^{\Nw\horizon} \rightarrow \real$, with $\bar{\phi}(\lambda, w) := \langle \lambda, \phi(w) \rangle$ is called a linear-convex incentive.
\end{definition}

\begin{remark}
\label{rem:linear-convex-incentive-componentwise}
    If $\phi$ is a vector function that is convex in each of its components, then $\phi$ is $\mathset{K}$-convex for the cone of nonnegative vectors, $\mathset{K} = \real^{\Nlambda}_+$.
    $\mathset{K}$ is self-dual, i.e., $\mathset{K}^* = \mathset{K} = \real^{\Nlambda}_+$.
\end{remark}

Since $\phi$ is a $\mathset{K}$-convex function, $\langle \lambda, \phi(\cdot) \rangle$ is a convex function for any given $\lambda \succeq_{\mathset{K}^*} 0$ (see Sec.~\ref{subsubsec:convex-vector-functions}).
So, for each $\lambda \succeq_{\mathset{K}^*} 0$, $\langle \lambda, \phi(\cdot) \rangle$ is a closed proper convex function and differentiable at each $w \in \real^{\Nw\horizon}$.
For the EV charging example in Sec.~\ref{subsec:motivating-example}, we choose $\mathset{K}$ as the set of nonnegative vectors $\real^{\Nlambda}_+$.
This corresponds to the fact that the unit price of electricity cannot be negative.
This nonnegativity constraint motivates the linear-convex incentives used in this section.


Following the outline of Sec.~\ref{sec:linear-incentives-for-hierarchical-mpc}, we define a dual function $\bar{g}^*: \mathset{K}^* \rightarrow \exreal$, similar to \eqref{eq:conv-analysis-conjugate}, as follows:
\begin{equation}
\label{eq:linear-convex-dual}
    \bar{g}^*(\lambda) := \textstyle{\min}_w \ (g(w) + \langle \lambda, \phi(w) \rangle).
\end{equation}
Note that $\bar{g}^*$ is similar to the dual function in \eqref{eq:iterative-method-proof-dual-function}.
Tab.~\ref{tab:definitions-for-symbols-used-in-the-paper} tabulates the symbols in the paper, which are used consistently through Secs.~\ref{sec:linear-incentives-for-hierarchical-mpc}-\ref{sec:multiple-lompcs}, and can be used as a quick reference.
The minimum in $w$ exists for each $\lambda$ because $\langle \lambda, \phi(\cdot)\rangle + g$ is $\constsc$-strongly convex.
Now we prove a result on the differentiability of $\bar{g}^*$, similar to \eqref{eq:conv-analysis-conjugate-subgrad}.




\begin{lemma} (Convexity and gradient of $-\bar{g}^*$)
\label{lem:linear-convex-incentive-conjugate-grad}
Under Assumption~\ref{assum:properties-of-gi}, Def.~\ref{def:linear-convex-incentive}, and \eqref{eq:linear-convex-dual}, $-\bar{g}^*$ is a proper convex function with $\dom{\bar{g}^*} = \mathset{K}^*$.
$-\bar{g}^*$ is differentiable for each $\lambda \in \mathset{K}^*$, with
\begin{equation}
\label{eq:linear-convex-incentive-conjugate-grad}
    \nabla (-\bar{g}^*)(\lambda) = -\phi(w^*(\lambda)), 
\end{equation}
where $w^*(\lambda)$ is the unique minimum of \eqref{eq:linear-convex-dual}, i.e.,
\begin{equation}
\label{eq:linear-convex-incentive-conjugate-grad-w}
    w^*(\lambda) = \textstyle{\argmin}_w \ (g(w) + \langle \lambda, \phi(w)\rangle).
\end{equation}
Further, $w^*(\cdot)$ is a continuous function on $\mathset{K}^*$.
\end{lemma}

\begin{proof}
The proof is provided in Appendix~\ref{subapp:proof-linear-convex-incentive-conjugate-grad}.
\end{proof}

Lem.~\ref{lem:linear-convex-incentive-conjugate-grad} shows that $-\bar{g}^*$ is a continuously differentiable convex function, with gradient given by $\nabla \bar{g}^*(\lambda) = \phi(w^*(\lambda))$.
Similar to Prop.~\ref{prop:duality}, the following result shows that $-\bar{g}^*$ is a Lipschitz smooth function.
We use this result in Thm.~\ref{thm:linear-convex-incentives-iterative-method} to show the convergence of an iterative method for calculating optimal incentives, similar to Thm.~\ref{thm:iterative-method} (see Tab.~\ref{tab:properties-of-the-dual-functions-and-main-results}).

\begin{lemma} (Lipschitz smoothness of $-\bar{g}^*$)
\label{lem:linear-convex-lipschitz-smooth-dual}
    Under Assumption~\ref{assum:properties-of-gi}, Def.~\ref{def:linear-convex-incentive}, and \eqref{eq:linear-convex-dual}, $-\bar{g}^*$ satisfies the following inequality for any $\lambda, \bar{\lambda} \in \mathset{K}^*$:
    \begin{equation}
    \label{eq:linear-convex-lipschitz-smooth}
    \begin{split}
        \hspace{-5pt} -\bar{g}^*(\lambda) \leq -\bar{g}^*(\lambda; \bar{\lambda}) := & -\bar{g}^*(\bar{\lambda}) + \langle \lambda - \bar{\lambda}, -\phi(\bar{w})\rangle \\
        & + 1/(2\constsc) \lVert D \phi(\bar{w})^\top (\lambda - \bar{\lambda}) \rVert^2,
    \end{split}
    \end{equation}
    where $\bar{w} = w^*(\bar{\lambda}) = \argmin_w \{g(w) + \langle \bar{\lambda}, \phi(w)\rangle\}$, $\constsc$ is the modulus of strong convexity of $g$, \changes{and $D\phi$ is the Jacobian of the vector function $\phi$}.
   In particular, when restricted to any compact domain, $-\bar{g}^*$ is Lipschitz smooth.
\end{lemma}

\begin{proof}
The proof is provided in Appendix~\ref{subapp:proof-linear-convex-lipschitz-smooth-dual}.
\end{proof}

\begin{table}[tbp]
\footnotesize
\setlength\extrarowheight{2pt}
\caption{Definitions for symbols used in the paper}\label{tab:definitions-for-symbols-used-in-the-paper}
\begin{tabularx}{\columnwidth}{|c|ll|L|}
\hline
Symbol & \multicolumn{2}{|l|}{Definition} & Meaning \\
\hline
$g^i(w^i; \xi_0, u)$ & Sec.~\ref{sec:linear-incentives-for-hierarchical-mpc}: & \eqref{eq:reduced-ith-lompc} & LoMPC cost (also $g^i(w^i)$) \\
\hline
$g^i(w^i; \xi_0, u, \lambda)$ & Sec.~\ref{sec:linear-incentives-for-hierarchical-mpc}: & \eqref{eq:incentivized-lompc-cost}, \eqref{eq:incentivized-lompc} & Incentive LoMPC cost and \\[-2pt] 
$w^{i*}(\xi_0, u, \lambda)$ & Sec.~\ref{sec:linear-convex-incentives}: & \eqref{eq:linear-convex-incentivized-lompc} & optimal solution (also \\[-1.5pt]
& Sec.~\ref{sec:multiple-lompcs}: & \eqref{eq:parametric-form-gi} & $g^i(w^i; \lambda)$ and $w^{i*}(\lambda)$) \\
\hline
$(u^*, \hat{w}^*)(\xi_0)$ & Sec.~\ref{sec:multiple-lompcs}: & \eqref{eq:robust-bimpc} & Team-optimal solution \\[-2pt]
& & & for robust BiMPC \\
\hline
$\bar{g}^*$, $\tilde{g}^*$ & Sec.~\ref{sec:linear-incentives-for-hierarchical-mpc}: & \eqref{eq:iterative-method-proof-dual-function} & Concave dual functions \\[-2pt]
& Sec.~\ref{sec:linear-convex-incentives}: & \eqref{eq:linear-convex-dual}, \eqref{eq:linear-convex-dual-transformed} & (related to the conjugate \\[-2pt]
& Sec.~\ref{sec:multiple-lompcs}: & \eqref{eq:bounded-incentive-controllability-conjugate-g} & function $g^*$, \eqref{eq:conv-analysis-conjugate}) \\
\hline
$\bar{g}^*(\cdot \,; \bar{\lambda})$ & Sec.~\ref{sec:linear-convex-incentives}: & \eqref{eq:linear-convex-lipschitz-smooth}, \eqref{eq:linear-convex-dual-transformed-lower-bound} & Concave quadratic \\[-1.5pt]
$(\tilde{g}^*(\cdot \,; \bar{\lambda}))$ & & & under-approximation of $\bar{g}^*$ ($\tilde{g}^*$) at $\bar{\lambda}$ \\
\hline
\end{tabularx}
\end{table}

Lem.~\ref{lem:linear-convex-lipschitz-smooth-dual} shows that the dual function $-\bar{g}^*$ can be majorized, i.e., bounded from above, by a quadratic function $-\bar{g}^*(\cdot \: ; \bar{\lambda})$ anchored at $\bar{\lambda}$.
Lems.~\ref{lem:linear-convex-incentive-conjugate-grad} and \ref{lem:linear-convex-lipschitz-smooth-dual} discuss the properties of the dual function $\bar{g}^*$.
In the rest of the section, we consider the incentive controllability property of the incentive hierarchical MPC and show that an iterative method can be used to compute an optimal incentive.
\hidetext{As noted in Rem.~\ref{rem:incentive-controllability-for-general-incentives}, incentive controllability does not necessarily hold for linear-convex incentives. (remove text ahead)}
Unlike for linear incentives, incentive controllability does not necessarily hold for linear-convex incentives.
From Def.~\ref{def:incentive-controllability}, incentive controllability holds if for all $w^* \in \dom{g}$, $\exists \lambda^* \in \mathset{K}^*$ such that $w^* = \argmin_w \{g(w) + \langle \lambda^*, \phi(w)\rangle\}$.
Using the optimality condition for $g(w) + \langle \lambda^*, \phi(w)\rangle$~\cite[Thm.~23.8]{rockafellar1997convex}, $w^*$ must satisfy
\begin{equation}
\label{eq:linear-convex-incentive-controllability-condition}
    -D \phi(w^*)^\top \lambda^* \in \partial g(w^*),
\end{equation}
for some $\lambda^* \in \mathset{K}^*$.
Using the above condition, we make the following assumption on the incentive structure.

\begin{assumption} (Linear-convex incentive controllability)
\label{assum:linear-convex-incentive-controllability}
$\phi$ satisfies the following properties:
\begin{enumerate}[leftmargin=*]
    \item $\phi(w^1) \preceq_{\mathset{K}} \phi(w^2) \Rightarrow w^1 = w^2, \ \forall w^1, w^2$.

    \item $\forall w^* \in \dom{g}$, $\exists \lambda^* \in \mathset{K}^*$, such that \eqref{eq:linear-convex-incentive-controllability-condition} holds.
\end{enumerate}
\end{assumption}

Following the discussion above, incentive controllability holds for linear-convex incentives when Assumption~\ref{assum:linear-convex-incentive-controllability} holds.
In the following remark, we identify a family of functions $\phi$ that satisfy Assumption~\ref{assum:linear-convex-incentive-controllability}.

\begin{remark}
\label{rem:linear-convex-incentive-controllability}
One class of linear-convex incentives that always satisfies Assumption~\ref{assum:linear-convex-incentive-controllability} is of the form $\langle \lambda_l, w\rangle + \langle \lambda_c, \phi_c(w)\rangle$, where $\phi_c$ is continuously differentiable and $\mathset{K}_c$-convex for some closed convex cone $\mathset{K}_c$, $\lambda_l \in \real^{\Nw\horizon}$, $\lambda_c \in (\mathset{K}_c)^*$, $\lambda = (\lambda_l, \lambda_c)$, and $\phi(w) = (w, \phi_c(w))$.
Here, $\mathset{K} = \{\mathbb{0}_{\Nw\horizon}\} \times \mathset{K}_c$, and $\mathset{K}^* = \real^{\Nw\horizon} \times (\mathset{K}_c)^*$.
Consequently, adding a linear-convex incentive to a linear incentive retains the incentive controllability property, Def.~\ref{def:incentive-controllability}.
\end{remark}

We make the following observation about the dual function $-\bar{g}^*$.
Given $w^* \in \dom{g}$, we define the function $\tilde{g}^*$ as follows:
\begin{equation}
\label{eq:linear-convex-dual-transformed}
    \tilde{g}^*(\lambda) := \bar{g}^*(\lambda) - \langle\lambda, \phi(w^*) \rangle.
\end{equation}
If $\lambda$ is a critical point of $\tilde{g}^*$, then by Lem.~\ref{lem:linear-convex-incentive-conjugate-grad},
\begin{equation*}
    \nabla \tilde{g}^*(\lambda) = 0 \ \Rightarrow \ \phi(w^*(\lambda)) - \phi(w^*) = 0.
\end{equation*}
By Assumption~\ref{assum:linear-convex-incentive-controllability}, this implies that $w^*(\lambda) = w^*$, i.e., $\lambda$ is an optimal incentive for $w^*$.
Thus, we claim that minimizing $-\tilde{g}^*$ results in an optimal incentive.
By Lem.~\ref{lem:linear-convex-lipschitz-smooth-dual}, for any $\lambda, \bar{\lambda} \in \mathset{K}^*$,
\begin{equation}
\label{eq:linear-convex-dual-transformed-lower-bound}
    -\tilde{g}^*(\lambda) \leq -\tilde{g}^*(\lambda \, ; \bar{\lambda}) := -\bar{g}^*(\lambda \, ; \bar{\lambda}) + \langle \lambda, \phi(w^*) \rangle,
\end{equation}
where $\bar{g}^*(\cdot \, ; \bar{\lambda})$ is defined in \eqref{eq:linear-convex-lipschitz-smooth} (see Tab.~\ref{tab:definitions-for-symbols-used-in-the-paper} for a summary).
So, $-\tilde{g}^*$ is majorized by the quadratic function $-\tilde{g}^*(\cdot \, ;\bar{\lambda})$.
Thus, to minimize $-\tilde{g}^*$, we can follow an iterative method where we minimize the majorization $-\tilde{g}^*(\cdot \,; \lambda^{(k)})$ at the iterate $\lambda^{(k)}$ to obtain the next iterate $\lambda^{(k+1)}$.
The following theorem, one of this paper's main results, formalizes the above discussion and guarantees the convergence of the majorization-minimization procedure.

\begin{theorem} (Iterative method for linear-convex incentives)
\label{thm:linear-convex-incentives-iterative-method}
Let Assumptions~\ref{assum:properties-of-gi} and \ref{assum:linear-convex-incentive-controllability} hold, and $w^* \in \dom{g}$.
Let $\lambda^{(k)}$ and $w^{(k)}$ be iterates satisfying
\begin{subequations}
\label{eq:linear-convex-incentives-iterative-method-const-step}
\begin{flalign}
    & w^{(k)} = w^*(\lambda^{(k)}) = \textstyle{\argmin}_w \, (g(w) + \langle \lambda^{(k)}, \phi(w)\rangle), \span\span \hspace{-5pt} \label{subeq:linear-convex-incentives-iterative-method-const-step-gradient} \\
    & \lambda^{(k+1)} = \textstyle{\argmin}_{\lambda \in \mathset{K}^*} \, \bigl(\epsilon^{(k)} \lVert \lambda - \lambda^{(k)} \rVert^2 -\tilde{g}^*(\lambda; \lambda^{(k)}) \bigr), \hspace{-15pt} && \label{subeq:linear-convex-incentives-iterative-method-const-step-flow}
\end{flalign}
\end{subequations}
where $\tilde{g}^*(\lambda; \lambda^{(k)}) := \bar{g}^*(\lambda; \lambda^{(k)}) - \langle \lambda, \phi(w^*)\rangle$, $\bar{g}^*(\cdot \, ; \lambda^{(k)})$ is the quadratic function defined in \eqref{eq:linear-convex-lipschitz-smooth}, and $\epsilon^{(k)} > 0$ is a bounded regularization weight.
If $\{\lambda^{(k)} \}$ is a bounded sequence, then $\lim_{k \rightarrow \infty} w^{(k)} = w^*$.

At each iteration, there is a guaranteed dual cost decrease:
\begin{align}
\label{eq:linear-convex-incentives-iterative-method-dual-cost-decrease}
     \tilde{g}^*(\lambda^{(k)}) - \tilde{g}^*(\lambda^{(k+1)}) \leq & -\tilde{g}^*(\lambda^{(k+1)}; \lambda^{(k)}) + \tilde{g}^*(\lambda^{(k)}; \lambda^{(k)}) \nonumber \\
    & + \epsilon^{(k)} \lVert \lambda^{(k+1)} - \lambda^{(k)}\rVert^2,
\end{align}
where $\tilde{g}^*(\lambda) := \bar{g}^*(\lambda) - \langle \lambda, \phi(w^*)\rangle$.

For any optimal incentive $\lambda^*$, and $\bar{\lambda} \in \mathset{K}^*$ such that $D\phi(w^*)^\top (\bar{\lambda} - \lambda^*) = 0$, $\bar{\lambda}$ is also an optimal incentive.
\end{theorem}

\begin{proof}
The proof is provided in Appendix~\ref{subapp:proof-linear-convex-incentives-iterative-method}.
\end{proof}

Similar to the iterative method for linear incentives, Thm.~\ref{thm:iterative-method}, the sequence $\{\lambda^{(k)}\}$ obtained from \eqref{eq:linear-convex-incentives-iterative-method-const-step} need not converge.
However, given an optimal incentive $\lambda^*$ (obtained as a limit point of $\{\lambda^{(k)}\}$), Thm.~\ref{thm:linear-convex-incentives-iterative-method} provides us a method to find other optimal incentives.
If $\bar{\lambda} \in \mathset{K}^*$ is such that $D\phi(w^*)^\top (\bar{\lambda} - \lambda^*) = 0$, then $\bar{\lambda}$ is also an optimal incentive.
As noted in the following remark, we can use this property to find an optimal incentive that minimizes some cost function.

\begin{remark} (Optimal incentive regularization)
\label{rem:optimal-incentive-regularization}
Let $\lambda^*$ be any optimal incentive obtained using ~\eqref{eq:linear-convex-incentives-iterative-method-const-step} and $c \in \real^{\Nlambda}$.
Then any solution $\bar{\lambda}$ to the conic program
\begin{equation}
\label{eq:optimal-incentive-regularization}
\begin{split}
    \min_{\lambda \in \mathset{K}^*} \ & \langle \lambda, c \rangle, \\
    \text{s.t} \quad & D\phi(w^*(\lambda^*))^\top (\lambda - \lambda^*) = 0,
\end{split}
\end{equation}
is also an optimal incentive.
Eq.~\eqref{eq:optimal-incentive-regularization} can be used to regularize along the set of optimal incentives.
For instance, in the EV charging example considered in Sec.~\ref{subsec:motivating-example}, the incentive $\lambda$ denotes the unit price of electricity, and we choose $c = \phi(w^*(\lambda^*))$ to minimize the total price paid by the EVs.
\end{remark}

\begin{remark}
Following Rem.~\ref{rem:linear-convex-incentive-controllability}, if we choose $\phi(w) = w$ and $\mathset{K} = \{\mathbb{0}_{\Nw\horizon}\}$, the linear-convex incentive $\langle \lambda, \phi(w)\rangle$ is the same as the linear incentive $\langle \lambda, w\rangle$.
In this case, Lem.~\ref{lem:linear-convex-incentive-conjugate-grad}, Lem.~\ref{lem:linear-convex-lipschitz-smooth-dual}, and Thm.~\ref{thm:linear-convex-incentives-iterative-method} are equivalent to \eqref{eq:conv-analysis-conjugate-subgrad}, Prop.~\ref{prop:duality}, and Thm.~\ref{thm:iterative-method} respectively (see Tab.~\ref{tab:properties-of-the-dual-functions-and-main-results} for a comparison).
\end{remark}

\hidetext{
TODO: subsection on equivalence between dual ascent and our algorithm, from the perspective of incentive compatibility
}

To summarize the results in this section, we first define a linear-convex incentive structure in Def.~\ref{def:linear-convex-incentive} and a corresponding dual function in \eqref{eq:linear-convex-dual}.
We prove the convexity properties of the dual function in Lems.~\ref{lem:linear-convex-incentive-conjugate-grad} and \ref{lem:linear-convex-lipschitz-smooth-dual}.
Next, we show that the hierarchical MPC problem, with the incentive LoMPC \eqref{eq:linear-convex-incentivized-lompc}, is incentive controllable when Assumption~\ref{assum:linear-convex-incentive-controllability} holds.
Finally, Thm.~\ref{thm:linear-convex-incentives-iterative-method} shows that for any team-optimal solution $\nu^* = (u^*, w^*)$, an optimal incentive $\lambda^*$ can be found iteratively, without the knowledge of the LoMPC cost function $g$.

In the following section, we consider the hierarchical MPC problem for the case of multiple LoMPCs (i.e., when $M > 1$), as defined in Sec.~\ref{subsec:hierarchical-formulation}.
We use the tools developed in Secs.~\ref{sec:linear-incentives-for-hierarchical-mpc} and \ref{sec:linear-convex-incentives} to propose an algorithm that solves a particular case of the incentive hierarchical MPC problem.


\section{Multiple Lower-level MPCs}
\label{sec:multiple-lompcs}

\subsection{Robust BiMPC Formulation}
\label{subsec:robust-bimpc-formulation}

In this section, we consider the hierarchical MPC problem \eqref{eq:reduced-ith-lompc} and \eqref{eq:reduced-bimpc} for the case of multiple LoMPCs.
Recall that the LoMPC problems are indexed by the set $\mathset{M}$ (with $M = |\mathset{M}|$).
First, we define the incentive BiMPC problem as an extension to \eqref{eq:reduced-bimpc} as follows (see Tab.~\ref{tab:hierarchical-mpc-variable-definitions} for variable definitions):
{
\begin{subequations}
\label{eq:multiple-lompcs-incentivized-bimpc}
\begin{alignat}{6}
    (\nu^*, \lambda^*)(\xi_0) \in \ && \argmin_{\substack{\nu=(u, \bm{w}), \\ \lambda \in \mathset{K}^*}} \ & f(u, w; \xi_0), \span\span\span\span\span \label{subeq:multiple-lompcs-incentivized-bimpc-cost} \\
    && \ \text{s.t} \quad &
    w^i &&= w^{i*}(\xi_0, u, \lambda), \ i \in \mathset{M}, \label{subeq:multiple-lompcs-incentivized-bimpc-w-optimality} \\
    &&& w &&= \frac{1}{M} \sum_{i \in \mathset{M}} w^i,
\end{alignat}
\end{subequations}
}%
where $w^{i*}(\xi_0, u, \lambda) := \argmin_{w^i} g^i(w^i; \xi_0, u, \lambda)$.
From \eqref{subeq:multiple-lompcs-incentivized-bimpc-w-optimality}, we note that all LoMPCs are provided the same incentive $\lambda \in \mathset{K}^*$.
We assume an affine parametric form for the incentive LoMPC cost functions $g^i$:
\begin{equation}
\label{eq:parametric-form-gi}
\begin{split}
    g^i(w^i; \xi_0, u, \lambda) := & \ g(w^i; \xi_0, u) + \langle \lambda, \phi(w^i)\rangle \\
    & + \langle \theta^i(\xi_0, u), w^i\rangle,
\end{split}
\end{equation}
where $\theta^i(\xi_0, u)$ is a parameter for the $i$-th LoMPC satisfying $\lVert \theta^i(\xi_0, u) \rVert \leq \bar{\theta}(\xi_0) \ \forall i \in \mathset{M}$, and $\langle \lambda, \phi(w)\rangle$ is a linear-convex incentive (see Def.~\ref{def:linear-convex-incentive}) with a $\mathset{K}$-convex function $\phi$ and $\lambda \in \mathset{K}^*$.
It can be shown that all LQR-like LoMPC problems have a parametric cost as given in \eqref{eq:parametric-form-gi} (see Rem.~\ref{rem:multiple-lompcs-lompc-lqr}).
In particular, the EV charging costs generally considered in literature can be written in the form \eqref{eq:parametric-form-gi} (see Sec.~\ref{subsec:ev-lompc-formulation}).
Note that for the single LoMPC case considered in Secs.~\ref{sec:linear-incentives-for-hierarchical-mpc} and \ref{sec:linear-convex-incentives}, we can choose $\theta^i \equiv 0$.

\begin{assumption}
\label{assum:multiple-lompcs-properties-of-gi}
For all $i \in \mathset{M}$, $g^i$ has the parametric form given by \eqref{eq:parametric-form-gi}, where the cost function $g$ satisfies Assumption~\ref{assum:properties-of-gi} and $\lVert \theta^i(\xi_0, u) \rVert \leq \bar{\theta}(\xi_0) \ \forall i \in \mathset{M}$.
\end{assumption}

\begin{remark}
\label{rem:multiple-lompcs-no-state-constraints}
From Sec.~\ref{subsec:hierarchical-formulation}, the cost function $g^i$ includes the state and input constraints for the $i$-th LoMPC.
If Assumption~\ref{assum:multiple-lompcs-properties-of-gi} holds, then, from \eqref{eq:parametric-form-gi}, $\dom{g^i}(\cdot \, ; \xi_0, u) = \dom{g^j}(\cdot \, ; \xi_0, u)$ for all $i, j \in \mathset{M}$ and $(\xi_0, u)$.
Therefore, to satisfy Assumption~\ref{assum:multiple-lompcs-properties-of-gi}, the LoMPCs in $\mathset{M}$ must not have any state constraints.
\end{remark}

Given Rem.~\ref{rem:multiple-lompcs-no-state-constraints}, the following remark identifies a class of LoMPCs for which Assumption~\ref{assum:multiple-lompcs-properties-of-gi} holds.

\begin{remark}
\label{rem:multiple-lompcs-lompc-lqr}
One class of LoMPC problems that always satisfy Assumption~\ref{assum:multiple-lompcs-properties-of-gi} is input-constrained LQRs with quadratic state cost and strongly convex input cost.
Any such LQR problem can be expressed in batch form as \cite[Sec.~8.2]{borrelli2017predictive}
\begin{equation}
\label{eq:input-constrained-lqr}
    \textstyle{\min}_{w^i \in \mathset{W}^\horizon} \bigl( (w^i)^\top Q w^i + \textstyle{\sum}_k \; g_w(w^i_k) + (y^i_0)^\top R w^i \bigr),
\end{equation}
where $Q \succeq 0$, $\mathset{W} \subset \real^\Nw$ is a closed convex set \changes{denoting the input constraints}, and $g_w: \real^\Nw \rightarrow \real$ is a strongly convex function.
In this case, $\theta^i(\xi_0, u) = R^\top y^i_0$.
The EV LoMPC considered in Sec.~\ref{sec:numerical-examples} is one such example.
\end{remark}

For the rest of the section, we suppress the dependence of $\theta$ and $\bar{\theta}$ on $(\xi_0, u)$.
Since all LoMPCs are provided the same incentive $\lambda \in \real^\Nlambda$ (from \eqref{subeq:multiple-lompcs-incentivized-bimpc-w-optimality}), incentive controllability is not guaranteed for the BiMPC problem \eqref{eq:multiple-lompcs-incentivized-bimpc} with multiple LoMPCs.
However, we can show that under Assumption~\ref{assum:multiple-lompcs-properties-of-gi}, incentive controllability holds with bounded error.

\begin{lemma} (Bounded incentive controllability)
\label{lem:bounded-incentive-controllability}
Let Assumptions~\ref{assum:linear-convex-incentive-controllability} and \ref{assum:multiple-lompcs-properties-of-gi} hold.
Given any $\xi_0$, let $(u^*, \hat{w}^*) \in \dom f$ with $\hat{w}^* \in \dom g$,
and $\lambda^* \in \mathset{K}^*$ be such that $-D\phi(\hat{w}^*)^\top\lambda^* \in \partial_w g(\hat{w}^*; \xi_0, u^*)$ (by Assumption~\ref{assum:linear-convex-incentive-controllability}).
Let, for $i \in \mathset{M}$, $w^{i*}$ be the optimal solution of the $i$-th LoMPC with the incentive $\lambda^*$:
\begin{equation}
\label{eq:bounded-incentive-controllability-wi}
    \hspace{-5pt} w^{i*} = \textstyle{\argmin}_{w^i} ( g(w^i; \xi_0, u^*) + \langle \lambda^*, \phi(w^i)\rangle + \langle \theta^i, w^i\rangle ). \hspace{-4pt}
\end{equation}
Then, $w^* = (1/M) \sum_{i \in \mathset{M}} w^{i*} \in \dom{g}$ and
\begin{equation}
\label{eq:bounded-incentive-controllability-error-bound}
    \bigl\lVert \hat{w}^* - w^* \bigr\rVert \leq \bar{\theta}/\constsc.
\end{equation}
\end{lemma}

\begin{proof}
Corresponding to the incentive LoMPC cost function $g^i$ in \eqref{eq:parametric-form-gi}, we define the dual function, $\bar{g}^*$, as follows:
\begin{equation}
\label{eq:bounded-incentive-controllability-conjugate-g}
    \bar{g}^*(\theta) := \textstyle{\min}_w \ ( g(w; \xi_0, u^*) + \langle \lambda^*, \phi(w) \rangle + \langle \theta, w\rangle ).
\end{equation}
Then, $\bar{g}^*(\theta) = -g^*(-\theta)$, where $g^*$ is the conjugate function (see Sec.~\ref{subsubsec:conjugate-of-a-convex-function}) of the $\constsc$-strongly convex function $g(\cdot \,; \xi_0, u^*) + \langle \lambda^*, \phi(\cdot) \rangle$.
So, by Prop.~\ref{prop:duality}, we have that $\bar{g}^*$ is a $(1/\constsc)$-smooth concave function.
By \eqref{eq:conv-analysis-conjugate-subgrad},
\begin{equation}
\label{eq:bounded-incentive-controllability-grad-g}
    \hspace{-4pt} \nabla \bar{g}^*(\theta) = \textstyle{\argmin}_w \, ( g(w; \xi_0, u^*) + \langle \lambda^*, \phi(w) \rangle + \langle \theta, w\rangle ).
\end{equation}
Since $-D\phi(\hat{w}^*)^\top\lambda^* \in \partial_w g(\hat{w}^*; \xi_0, u^*)$, we have from the optimality condition for $g(w; \xi_0, u^*) + \langle \lambda^*, \phi(w)\rangle$ that~\cite[Thm.~23.8]{rockafellar1997convex}
\begin{equation}
\label{eq:bounded-incentive-controllability-w-hat}
    \hat{w}^* = \textstyle{\argmin}_w \ ( g(w; \xi_0, u^*) + \langle \lambda^*, \phi(w) \rangle ).
\end{equation}
Using \eqref{eq:bounded-incentive-controllability-grad-g}, \eqref{eq:bounded-incentive-controllability-wi}, \eqref{eq:bounded-incentive-controllability-w-hat}, and the fact that $\bar{g}^*$ is $(1/\constsc)$-smooth, we get that
\begin{align*}
    \lVert \hat{w}^* - w^{i*}\rVert = \lVert \nabla \bar{g}^*(0) - \nabla \bar{g}^*(\theta^i) \rVert \leq \lVert \theta^i \rVert/\constsc \leq \bar{\theta}/\constsc.
\end{align*}
Finally, \eqref{eq:bounded-incentive-controllability-error-bound} follows from the triangle inequality.
Since $w^{i*} \in \dom{g} \ \forall i$ and $\dom{g}$ is convex, $w^* \in \dom{g}$.
\end{proof}

\hidetext{
\begin{remark}
TODO: In the absence of additional information about the LoMPC cost function $g$, the bound provided by \eqref{eq:bounded-incentive-controllability-error-bound} is tight.
To see this, we note that the error bound $\bar{\theta}/\constsc$ arises from the $(1/\constsc)$-Lipschitz smoothness of the dual function $\bar{g}^*$ in \eqref{eq:bounded-incentive-controllability-conjugate-g}.
We consider a $1$-dimensional example as follows:
Let $\theta^1 = -1$ and $\theta^i = 1$ for $2 \leq i \leq M$, with $\bar{\theta} = 1$ and $\lambda^* = 0$.
Let $g(w) = w^2 + \indicator{\{w \geq 0\}}$, where $\indicator{(\cdot)}$ is the indicator function (see Sec.~\ref{subsec:notations}), with $\constsc = 2$.
Then from \eqref{eq:bounded-incentive-controllability-w-hat}, $\hat{w}^* = 0$, and from \eqref{eq:bounded-incentive-controllability-wi}, $w^{1*} = 1/2$, $w$.
(See notebook for example).
\end{remark}
}

\hidetext{TODO: Remark about more complex parametric forms.}

For any $w \in \dom{g}$, Lem.~\ref{lem:bounded-incentive-controllability} proves the existence of an incentive $\lambda^*$ such that the average LoMPC solution $w^*$ with incentive $\lambda^*$ has bounded error compared to $w$.
Moreover, from \eqref{eq:bounded-incentive-controllability-w-hat}, $\lambda^*$ is an optimal incentive corresponding to $w$, when $\theta = 0$.
Therefore, using Thm.~\ref{thm:linear-convex-incentives-iterative-method}, we can iteratively compute $\lambda^*$ given $(u^*, w)$.
We also note that the error bound given by \eqref{eq:bounded-incentive-controllability-error-bound} is independent of the number of LoMPCs $M$.



In the case of a single LoMPC, as considered in Sec.~\ref{sec:linear-incentives-for-hierarchical-mpc}, the incentive controllability property for linear incentives leads to the equivalent formulation \eqref{eq:equivalent-bimpc} of the incentive BiMPC \eqref{eq:incentivized-bimpc}.
Similarly, we can show that the bounded incentive controllability property, Lem.~\ref{lem:bounded-incentive-controllability}, leads to a robust formulation of the incentive BiMPC \eqref{eq:multiple-lompcs-incentivized-bimpc}.
By Lem.~\ref{lem:bounded-incentive-controllability}, for any $w \in \dom{g}$, we can find an incentive such that the average LoMPC solution $w^* \in \dom{g}$ has bounded error compared to $w$.
We must account for this error to ensure that $(u^*, w^*)$ is also feasible for the incentive BiMPC \eqref{eq:multiple-lompcs-incentivized-bimpc}.

Next, we describe a robust formulation of \eqref{eq:multiple-lompcs-incentivized-bimpc} using Lem.~\ref{lem:bounded-incentive-controllability}.
The following main result is a generalization of Thm.~\ref{thm:bimpc-equivalence} for the case of multiple LoMPCs (see Tab.~\ref{tab:properties-of-the-dual-functions-and-main-results}).

\begin{theorem} (Robust BiMPC formulation)
\label{thm:robust-bimpc-formulation}
Let Assumptions~\ref{assum:linear-convex-incentive-controllability} and \ref{assum:multiple-lompcs-properties-of-gi} hold.
For any $\xi_0$, let $(u^*, \hat{w}^*)$ be an optimal solution of the following optimization problem:
\begin{subequations}
\label{eq:robust-bimpc}
\begin{alignat}{6}
    (u^*, \hat{w}^*) \in \ && \argmin_{\nu=(u, \hat{w})} \ & f(\nu; \xi_0), \label{subeq:robust-bimpc-cost} \\
    && \ \text{s.t} \quad &
    \hat{w} \in \dom g(\cdot \; ; \xi_0, u), \label{subeq:robust-bimpc-w-g-feasibility} \\
    &&& \{(u, \hat{w})\} + \mathset{E} \subset \dom f(\cdot \; ; \xi_0), \label{subeq:robust-bimpc-w-f-feasibility}
\end{alignat}
\end{subequations}
where the addition between sets in \eqref{subeq:robust-bimpc-w-f-feasibility} is the Minkowski sum (see Sec.~\ref{subsec:notations}), and
\begin{equation}
\label{eq:robust-bimpc-error-set}
    \mathset{E} := \{(\mathbb{0}, w) \in \real^{\Nu\horizon} \times \real^{\Nw\horizon}: \lVert w \rVert \leq \bar{\theta}/\constsc\}.
\end{equation}
Then, for $\lambda^* \in \! \mathset{K}^*$ such that $-D\phi(\hat{w}^*)^\top\lambda^* \in \partial_w g(\hat{w}^*; \xi_0, u^*)$,
$(u^*, \bm{w}^*, \lambda^*)$ is feasible for the incentive BiMPC \eqref{eq:multiple-lompcs-incentivized-bimpc} and
\begin{equation*}
    \bigl\lVert \hat{w}^* - w^* \bigr\rVert \leq \bar{\theta}/\constsc,
\end{equation*}
where $\bm{w}^* = (w^{1*}, ..., w^{M*})$, $w^* = (1/M) \sum_{i \in \mathset{M}} w^{i*}$, and
\begin{equation*}
    w^{i*} = \textstyle{\argmin}_{w^i} ( g(w^i; \xi_0, u^*) + \langle \lambda^*, \phi(w^i) \rangle + \langle \theta^i, w^i \rangle ).
\end{equation*}
\end{theorem}

\begin{proof}
Since $\mathbb{0}_{(\Nu+\Nw)\horizon} \in \mathset{E}$ and $(u^*, \hat{w}^*)$ satisfies the constraint \eqref{subeq:robust-bimpc-w-f-feasibility}, $(u^*, \hat{w}^*) \in \dom f(\cdot \; ; \xi_0)$.
The assumptions of Lem.~\ref{lem:bounded-incentive-controllability} are satisfied, and so \eqref{eq:bounded-incentive-controllability-error-bound} holds.
Thus, by \eqref{eq:bounded-incentive-controllability-error-bound} and \eqref{subeq:robust-bimpc-w-f-feasibility},
\begin{equation*}
    (u^*, w^*) \in (u^*, \hat{w}^*) + \mathset{E} \subset \dom f(\cdot \; ; \xi_0).
\end{equation*}
Therefore, $(u^*, \bm{w}^*, \lambda^*)$ is feasible for the incentive BiMPC \eqref{eq:multiple-lompcs-incentivized-bimpc}.
\end{proof}

The robust BiMPC \eqref{eq:robust-bimpc} defines a team-optimal solution $(u^*, \hat{w}^*)$ for the incentive hierarchical problem given by \eqref{eq:multiple-lompcs-incentivized-bimpc} and \eqref{eq:linear-convex-incentivized-lompc}.
Since the same incentive is provided to all LoMPCs, it might not be possible to obtain the team-optimal solution $(u^*, \hat{w}^*)$ for any incentive.
However, Thm.~\ref{thm:robust-bimpc-formulation} shows that we can achieve a solution $(u^*, w^*)$ that has bounded error compared to the team-optimal solution.
The robust formulation \eqref{eq:robust-bimpc} shrinks the domain of the BiMPC problem to account for this error and thus guarantees that the obtained solution $(u^*, w^*)$ satisfies the BiMPC constraints.

\hidetext{
\begin{remark} (Bounded optimality of the Robust BiMPC)
The error bound $\lVert \hat{w}^* - w^* \rVert \leq \bar{\theta}/\constsc$ from Thm.~\ref{thm:robust-bimpc-formulation}, together with the optimal value of \eqref{eq:robust-bimpc}, also provides an upper bound on the optimal value of the incentive BiMPC \eqref{eq:multiple-lompcs-incentivized-bimpc}.
This optimality bound can be used to guarantee stability (in the sense of Lyapunov) for the system.
\end{remark}
}

\begin{remark} (Explicit robust BiMPC formulation)
\label{rem:explicit-robust-bimpc-formulation}
Since the state and input constraints of the BiMPC are convex (see Sec.~\ref{subsec:hierarchical-formulation}), $\dom{f}$ is convex.
Thus, the robustness constraint \eqref{subeq:robust-bimpc-w-f-feasibility} is a convex constraint.
When the state constraints of the BiMPC problem \eqref{eq:bimpc} are polytopic or quadratic, we can explicitly write the robustness constraint \eqref{subeq:robust-bimpc-w-f-feasibility} using duality theory.
We show this explicit formulation in App.~\ref{app:robust-bimpc-formulation}.
\end{remark}

\begin{remark} (Robustness horizon)
\label{rem:robustness-horizon}
The robust BiMPC \eqref{eq:robust-bimpc} accounts for the error in the optimization variable $\hat{w}$ for the entire horizon $\horizon$.
The robustness horizon $\horizon_r > 0$ can be defined as the number of time steps for which the error in $\hat{w}$ is considered.
For $\horizon_r \leq \horizon$, we modify the robust BiMPC \eqref{eq:robust-bimpc} by redefining $\mathset{E}$ in \eqref{subeq:robust-bimpc-w-f-feasibility} as follows:
\begin{equation*}
    \mathset{E} := \{(\mathbb{0}, w_r, \mathbb{0}) \in \real^{\Nu\horizon} \times \real^{\Nw\horizon_r} \times \real^{\Nw(\horizon-\horizon_r)}: \lVert w_r \rVert \leq k \bar{\theta}/\constsc\},
\end{equation*}
where $k = \max_w \{\lVert w_{0:\horizon_r-1} \rVert: \lVert w\rVert \leq 1\}$.
Using $\horizon_r < \horizon$ expands the feasible set of \eqref{eq:robust-bimpc} but only guarantees the satisfaction of the BiMPC constraints for $\horizon_r$ time steps.
In other words, the state $x^*_k$ obtained from $(u^*, \bm{w}^*)$ might violate the BiMPC state constraints for $k > \horizon_r$.
However, since the hierarchical MPC problem is solved with state feedback at each time step (see Fig.~\ref{fig:hierarchical-mpc-formulation}), we can choose $N_r < N$, as long as \eqref{eq:robust-bimpc} is feasible.
\end{remark}

\begin{remark}
\label{rem:partition-of-lompcs}
    We can extend the incentive BiMPC \eqref{eq:multiple-lompcs-incentivized-bimpc} to account for a partition of the set of LoMPCs $\mathset{M}$ as $\sqcup_j \, \mathset{M}_j$.
    Here, the BiMPC assigns a different incentive $\lambda^j$ to each set of LoMPCs $\mathset{M}_j$.
    We can formulate a robust BiMPC similar to \eqref{eq:robust-bimpc} by considering the error bounds corresponding to each partition.
\end{remark}

To summarize the results in this section, we first specify the assumptions on the incentive hierarchical MPC problem in Assumption~\ref{assum:multiple-lompcs-properties-of-gi}.
Then, we show that although incentive controllability does not hold for the case of multiple LoMPCs, we can still guarantee bounded incentive controllability, Lem.~\ref{lem:bounded-incentive-controllability}.
Finally, we use the bounded incentive controllability property to formulate a robust BiMPC, Thm.~\ref{thm:robust-bimpc-formulation}, which defines a team-optimal solution.
In the next subsection, we outline a procedure for solving the hierarchical MPC problem given by \eqref{eq:multiple-lompcs-incentivized-bimpc} and \eqref{eq:linear-convex-incentivized-lompc} using the results from Secs.~\ref{sec:linear-convex-incentives} and \ref{sec:multiple-lompcs}.

\subsection{Computation of the Hierarchical MPC Solution}
\label{subsec:computation-of-the-hierarchical-mpc-solution}

\begin{figure}[tp]
    \centering
    \includegraphics[width=0.99\columnwidth]{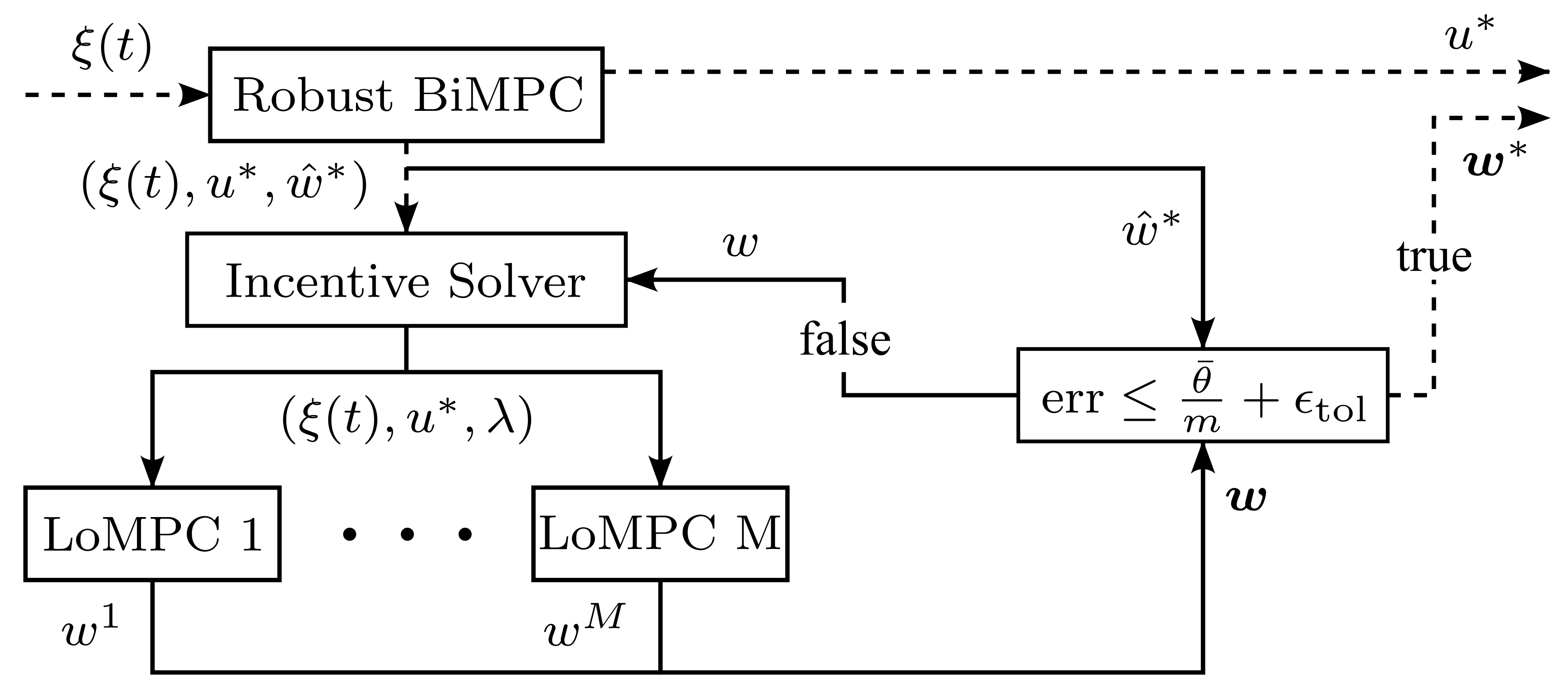}
    \caption{\changes{A framework for solving the hierarchical MPC problem given by \eqref{eq:multiple-lompcs-incentivized-bimpc} and \eqref{eq:linear-convex-incentivized-lompc} (this figure corresponds to the controller subsystem in Fig.~\ref{fig:hierarchical-mpc-formulation})}.
    At each time step, given the current state $\xi(t)$, we first solve the robust BiMPC \eqref{eq:robust-bimpc} to get the team-optimal solution $(u^*, \hat{w}^*)$.
    By the bounded incentive controllability property, Lem.~\ref{lem:bounded-incentive-controllability}, we can find an incentive $\lambda^*$ such that the average LoMPC solution $w^*$ has bounded error compared to $\hat{w}^*$.
    The \emph{incentive solver} uses the iterative method in Thm.~\ref{thm:linear-convex-incentives-iterative-method} to compute $\lambda^*$.
    The output $(u^*, \bm{w}^*)$ is guaranteed to be feasible for the incentive BiMPC \eqref{eq:multiple-lompcs-incentivized-bimpc}.
    In the flowchart, the solid lines indicate information flows that happen multiple times per time step, while the dashed lines indicate those that happen once.
    }
    \label{fig:hierarchical-mpc-solution-procedure}
\end{figure}

We propose a two-step procedure to solve the hierarchical MPC problem given by \eqref{eq:multiple-lompcs-incentivized-bimpc} and \eqref{eq:linear-convex-incentivized-lompc}.
The method is depicted in Fig.~\ref{fig:hierarchical-mpc-solution-procedure}, and Algs.~\ref{alg:iterative-method} and \ref{alg:incentivized-bimpc-solution} provide an algorithm for the solution computation.

We first describe the overall method, as shown in Fig.~\ref{fig:hierarchical-mpc-solution-procedure} and Alg.~\ref{alg:incentivized-bimpc-solution}.
At each time step, given the current state $\xi(t)$, we first solve the robust BiMPC problem \eqref{eq:robust-bimpc} to obtain a team-optimal solution $(u^*, \hat{w}^*)$ (line~\ref{alg-line:solve-robust-bimpc}, Alg.~\ref{alg:incentivized-bimpc-solution}).
The team-optimal solution is provided to an \emph{incentive solver} that iteratively computes an optimal incentive $\lambda^*$ corresponding to $\hat{w}^*$ using Thm.~\ref{thm:linear-convex-incentives-iterative-method} (line~\ref{alg-line:solve-optimal-incentive}, Alg.~\ref{alg:incentivized-bimpc-solution}).
We note that, from Thm.~\ref{thm:robust-bimpc-formulation}, the LoMPC solution $\bm{w}^*$ obtained using the incentive $\lambda^*$ is feasible for the incentive BiMPC \eqref{eq:multiple-lompcs-incentivized-bimpc}.
The optimal solution $(u^*, \bm{w}^*)$ is then used as an input to the dynamical system \eqref{eq:coupled-dynamics} to obtain the next state $\xi(t+1)$ (see Fig.~\ref{fig:hierarchical-mpc-formulation}).

Next, we describe the \emph{incentive solver}, which is used to compute an optimal incentive $\lambda^*$ given the team-optimal solution $(u^*, \hat{w}^*)$, as shown in Alg.~\ref{alg:iterative-method}.
We follow the iterative method described in Thm.~\ref{thm:linear-convex-incentives-iterative-method}.
At each iteration, we update the $w$-iterate using \eqref{subeq:linear-convex-incentives-iterative-method-const-step-gradient} to the average LoMPC solution obtained from the current incentive $\lambda$ (line~\ref{alg-line:w-update}, Alg.~\ref{alg:iterative-method}).
The $\lambda$-iterate is updated using the majorization-minimization step in \eqref{subeq:linear-convex-incentives-iterative-method-const-step-flow} (line~\ref{alg-line:lambda-update}, Alg.~\ref{alg:iterative-method}).
We iterate these update steps until the error $\lVert \hat{w}^* - w\rVert$ is less than the $\bar{\theta}/\constsc$, which is the error bound guaranteed by Lem.~\ref{lem:bounded-incentive-controllability}, up to a tolerance $\epsilon_\text{tol} > 0$.
We emphasize that the iterative method described by Alg.~\ref{alg:iterative-method} does not require the knowledge of the LoMPC functions $g^i$.
Instead, the optimal solutions from the LoMPCs are repeatedly queried to obtain an optimal incentive (see Assumption~\ref{assum:privacy-of-lompc}).

\begin{algorithm}[tp]
\small
\DontPrintSemicolon
\caption{Incentive BiMPC Solver}\label{alg:incentivized-bimpc-solution}
\Input{$\xi(t)$ \tcp*[r]{Current state}}
\Output{$(u^*, \bm{w}^*)$, $\lambda^*$ \tcp*[r]{Incentive BiMPC solution}}
\Function{
\solveIncentiveBiMPC{$\xi(t)$} \\
\tcp*[h]{Solves the incentive BiMPC problem \eqref{eq:multiple-lompcs-incentivized-bimpc} using \eqref{eq:robust-bimpc}}
}{
    $\xi_0 \gets \xi(t)$\;
    \tcp*[h]{Solve robust BiMPC \eqref{eq:robust-bimpc}}\;
    $(u^*, \hat{w}^*) \gets$ \solveRobustBiMPC{$\xi_0$}\label{alg-line:solve-robust-bimpc}\;
    \tcp*[h]{Optimal incentive from Alg.~\ref{alg:iterative-method}}\;
    $\lambda^* \gets$ \optimalIncentive{$\xi_0, u^*, \hat{w}^*$}\label{alg-line:solve-optimal-incentive}\;
    \For{$i \gets 1$ \KwTo $M$}{
    $w^{i*} \gets w^{i*}(\xi_0, u^*, \lambda^*)$ \tcp*[r]{see \eqref{subeq:multiple-lompcs-incentivized-bimpc-w-optimality}}
    }
    $\bm{w}^* \gets (w^{1*}, ..., w^{M*})$\;
    \KwRet $u^*$, $\bm{w}^*$, $\lambda^*$\;
}
\end{algorithm}

\subsection{Advantages and Disadvantages of Proposed Method}
\label{subsec:advantages-and-disadvantages-of-proposed-method}

Finally, we discuss the advantages and disadvantages of our method compared to an exact bilevel optimization solver, such as the one proposed in \cite{mintz2018control}.
The advantages of our method are as follows:
1) Our method scales well to a large number of LoMPCs.
In particular, neither the robust BiMPC computation time nor the optimal incentive computation time grows with the number of LoMPCs (the optimal incentive computation relies on the LoMPC solution for any incentive, which is performed in a distributed manner).
Also, the error bound provided by Lem.~\ref{lem:bounded-incentive-controllability} is independent of the number of LoMPCs.
In contrast, the BiMPC reformulation in \cite{mintz2018control} grows in size with the number of LoMPCs.
2) Our method only requires the knowledge of the LoMPC domain, but not the LoMPC cost function.
The disadvantages of our method are as follows:
1) Our method only guarantees the optimality of the incentive BiMPC problem \eqref{eq:multiple-lompcs-incentivized-bimpc} up to an error bound.
2) Our method places constraints on the LoMPC structure (see Assumption~\ref{assum:multiple-lompcs-properties-of-gi}) and does not provide any guarantees when the incentive $\lambda$ is constrained.
On the other hand, the BiMPC reformulation in \cite{mintz2018control} allows for constraints on the incentive $\lambda$ and computes the optimal solution of the BiMPC.

\begin{algorithm}[tp]
\small
\DontPrintSemicolon
\caption{Iterative Method for Optimal Incentives}\label{alg:iterative-method}
\Input{$\xi_0$, $(u^*, \hat{w}^*)$ \tcp*[r]{Initial state, team-optimal solution}}
\Parameter{$\epsilon_{\text{tol}}$, $\bar{\theta}$, $\constsc$}
\Output{$\lambda^*$ \tcp*[r]{$\lambda^*$ is such that \eqref{eq:bounded-incentive-controllability-wi}, \eqref{eq:bounded-incentive-controllability-error-bound} hold}}
\Function{\optimalIncentive{$\xi_0$, $u^*$, $\hat{w}^*$} \\
\tcp*[h]{Solves for an optimal incentive using Thm.~\ref{thm:linear-convex-incentives-iterative-method}}}{
    $\lambda \gets \lambda_{\text{ws}}$\tcp*[r]{Warm-start solution from previous call}
    $(w, \text{err}) \gets$ \LoMPCSolution{$\xi_0$, $u^*$, $\lambda$, $\hat{w}^*$}\;
    \While{$\text{err} > \bar{\theta}/\constsc + \epsilon_{\text{tol}}$}{
    \tcp*[h]{Solve \eqref{subeq:linear-convex-incentives-iterative-method-const-step-flow} with $(w^{(k)}, \lambda^{(k)}, w^*) = (w, \lambda, \hat{w}^*)$}\;
    $\lambda \gets$ \incentiveUpdate{$w$, $\lambda$, $\hat{w}^*$}\label{alg-line:lambda-update}\;
    \tcp*[h]{Update $w$, see \eqref{subeq:linear-convex-incentives-iterative-method-const-step-gradient}}\;
    $(w, \text{err}) \gets$ \LoMPCSolution{$\xi_0$, $u^*$, $\lambda$, $\hat{w}^*$}\label{alg-line:w-update}
    }
    \tcp*[h]{Regularize incentive using \eqref{eq:optimal-incentive-regularization}, Rem.~\ref{rem:optimal-incentive-regularization}}\;
    $\lambda \gets$ \incentiveRegularization{$\lambda$, $w$}\;
    \KwRet $\lambda$
}
\Function{\LoMPCSolution{$\xi_0$, $u^*$, $\lambda$, $\hat{w}^*$}}{
    \tcp*[h]{Average LoMPC solution, see \eqref{eq:linear-convex-incentivized-lompc} and \eqref{eq:parametric-form-gi}}\;
    $w \gets$ $(1/M)\sum_{i \in \mathset{M}} w^{i*}(\xi_0, u^*, \lambda)$\;
    \tcp*[h]{Convergence error, see Lem.~\ref{lem:bounded-incentive-controllability}}\;
    $\text{err} \gets$ $\lVert \hat{w}^* - w\rVert$\;
    \KwRet $(w, \text{err})$
}
\end{algorithm}

\section{Numerical Example: Dynamic Price Control for EV Charging}
\label{sec:numerical-examples}

In this section, we consider dynamic price control for electric vehicle (EV) charging (see Sec.~\ref{subsec:motivating-example}).
We mostly use the dynamical systems and costs in \cite{ma2013decentralized,zou2017efficient}, but specify when we deviate.
The set of MPC problems solved by the EVs comprises the set of LoMPCs.
The optimization problem solved by the ISO comprises the BiMPC problem.
Together, the BiMPC and the LoMPC problems form an incentive hierarchical MPC problem, where the incentive provided by the ISO is the unit price of electricity.
The Python code for the example can be found in the repository\footnote{
\label{code}
The implementation code can be found at \url{https://github.com/AkshayThiru/incentive-design-mpc}.
}.
Next, we describe the LoMPC and BiMPC problems in detail.

\subsection{LoMPC Formulation}
\label{subsec:ev-lompc-formulation}

The $i$-th EV has a state variable $y^i_k \in \real$, which denotes its battery state of charge (SoC) as a fraction of its total capacity $\Theta^i$ at time step $k$.
$y^i_k$ is called the normalized SoC of EV $i$.
The fraction of the EV battery capacity charged during time interval $k$ is given by $w^i_k \in \real$.
$w^i_k$ is called the normalized charging rate at time step $k$.
In addition to the system considered in \cite{zou2017efficient}, we assume that $w^i_k$ is bounded by $w^i_\text{max}$.
While a more realistic model of battery charging might consider $w^i_\text{max}$ as a function of the SoC, we ignore this complication.
The normalized SoC of EV $i$ evolves according to the dynamics:
\begin{equation}
\label{eq:ev-lompc-dynamics}
\begin{gathered}
    y^i_{k+1} = y^i_k + w^i_k, \\
    0 \leq w^i_k \leq w^i_\text{max} \leq 1, \quad 0 \leq y^i_k \leq y^i_\text{max} \leq 1.
\end{gathered}
\end{equation}
We assume that the EV population comprises $500$ small and large EVs for a total population size of $M = 1000$.
Additionally, all the small EVs have the same values for $\Theta, w_\text{max}$, and $y_\text{max}$, given by $\Theta^s, w^s_\text{max}$, and $y^s_\text{max}$, respectively.
Likewise, the large EVs share the values $\Theta^l, w^l_\text{max}$, and $y^l_\text{max}$.
While this is a simplifying assumption compared to \cite{zou2017efficient}, we note that \cite{zou2017efficient} sets different prices for each EV.
On the other hand, we set identical electricity prices for large parts of the EV population.
In the rest of the subsection, we discuss the small EV LoMPC formulation; a similar formulation is used for the large EVs.

Each small EV $i$ solves a LoMPC problem with horizon $\horizon$ to determine the amount of electricity it draws during a time interval.
Following the notation in Tab.~\ref{tab:hierarchical-mpc-variable-definitions}, let $y^i = (y^i_0, ..., y^i_{\horizon})$ and $w^i = (w^i_0, ..., w^i_{\horizon-1})$.
The dynamics and constraints of small EV $i$ are given by \eqref{eq:ev-lompc-dynamics} with $w^i_\text{max} = w^s_\text{max}$ and $y^i_\text{max} = y^s_\text{max}$.
Without loss of generality, we ignore the state constraints $\mathbb{0} \leq y^i \! \leq y^s_\text{max} \mathbb{1}$ by assuming that EVs stop charging once their maximum normalized SoC $y^s_\text{max}$ is attained.
We also assume, for simplicity, that once an EV is fully charged, it is replaced by another EV with a random normalized SoC.
We also define the total battery capacity of the EV population as follows:
\begin{equation}
\label{eq:ev-lompc-b}
    B = (M/2) \Theta^s + (M/2) \Theta^l.
\end{equation}

The LoMPC of each EV has a cost function with three components, which we describe next.

\subsubsection{Battery Degradation Cost}

The charging behavior over time affects an EV battery's health.
We associate a battery degradation cost function $g^i_\text{bat}$, a function of the normalized charging rate $w^i_k$, to the rate of decline of the battery health \cite{zou2017efficient}.
The total battery degradation cost is given by
\begin{equation}
\label{eq:ev-lompc-battery-degradation-cost}
    (\Theta^i)^2 \, \textstyle{\sum}_{k=0}^{\horizon-1} \, g^i_\text{bat}(w^i_k),
\end{equation}
where the $(\Theta^i)^2$ factor normalizes the battery degradation cost term with respect to the other cost terms.
We assume that $g^i_\text{bat}$ can be any closed convex function (possibly nonsmooth) and that $g^i_\text{bat}$ is unknown to the ISO.
Additionally, all small and large EVs share the battery degradation cost functions $g^s_\text{bat}$ and $g^l_\text{bat}$, respectively (the exact forms can be found in the code\footref{code}).

\subsubsection{Charging Cost}

Each EV $i$ has a quadratic tracking cost at time $k$ given by $\delta (\Theta^i)^2(y^i_\text{max} - y^i_k)^2$, where $\delta > 0$ is a fixed parameter.
This cost function imposes a penalty for not being at the desired SoC.
The total charging cost is given by
\begin{equation}
\label{eq:ev-lompc-tracking-cost}
    \delta (\Theta^i)^2 \textstyle{\sum}_{k=1}^N \, (y^i_\text{max} - y^i_k)^2.
\end{equation}

\subsubsection{Electricity Cost}

While \cite{zou2017efficient} considers an auction mechanism, we assume that the ISO sets identical electricity prices for parts of the EV population.
Since electricity prices are provided by the ISO as incentives to the EVs, the electricity cost is the incentive cost term, as defined in \eqref{eq:linear-convex-incentivized-lompc}.
The total electricity cost for EV $i$ is given by
\begin{equation}
\label{eq:ev-lompc-electricity-price-cost}
    \Theta^i \bigl(\lambda_{l_1}^\top w^i + \lambda_{l_2}^\top (w^i_\text{max} \mathbb{1} - w^i) + (w^i)^\top \text{diag}(\lambda_q) w^i \bigr),
\end{equation}
where $\lambda_{l_1}, \lambda_{l_2}, \lambda_q \in \real^\horizon_+$
and $\lambda = (\lambda_{l_1}, \lambda_{l_2}, \lambda_q) \in \real^{3\horizon}_+$ is the incentive.
The first term in \eqref{eq:ev-lompc-electricity-price-cost} penalizes higher electricity consumption, while the third term adds a nonlinear electricity cost.
The second term reduces prices for higher consumption.
This can be interpreted as the ISO providing compensation to the EVs for charging at a higher rate at the expense of EV battery degradation.
The electricity cost for small EV $i$ can be expressed as $\langle \lambda, \phi^s(w^i)\rangle$, where 
\begin{equation}
\label{eq:ev-phi-function}
    \phi^s(w^i) = \Theta^s \bigl(w^i, w^s_\text{max} \mathbb{1} - w^i, w^i \odot w^i\bigr),
\end{equation}
where $\odot$ represents the Hadamard product.
From \eqref{eq:ev-phi-function} and Rem.~\ref{rem:linear-convex-incentive-componentwise}, $\phi^s$ is $\mathset{K}$-convex with $\mathset{K} = \real^{3\horizon}_+$, and $\lambda \in \mathset{K}^* = \real^{3\horizon}_+$.
Thus, $\langle \lambda, \phi^s(w^i)\rangle$ is a linear-convex incentive (see Def.~\ref{def:linear-convex-incentive}).

The LoMPC problem for small EV $i$ is given by the dynamics \eqref{eq:ev-lompc-dynamics} (without the state constraints) and the three cost terms discussed above.
We can verify that the small EV LoMPC matches the form in \eqref{eq:input-constrained-lqr}.
Thus, by Rem.~\ref{rem:multiple-lompcs-lompc-lqr}, the small EV LoMPC satisfies Assumption~\ref{assum:multiple-lompcs-properties-of-gi} and can be written in the parametric form \eqref{eq:parametric-form-gi} with
\begin{equation}
\begin{gathered}
    \theta^i(\bm{y}_0) = 2\delta (\Theta^s)^2 \bigl( y^s_{0, c} - y^i_0 \bigr) \mathbb{1}, \\
    y^s_{0, c} = \frac{1}{2} \biggl(\max_{i \in \text{small EVs}} \{y^i_0\} + \min_{i \in \text{small EVs}} \{y^i_0\}\biggr).
\end{gathered}
\end{equation}
Similarly, we can define the LoMPC problem for a large EV.

\subsection{LoMPC Solution Error Bounds}
\label{subsec:ev-lompc-solution-error-bounds}

\begin{figure}
    \centering
    \includegraphics[width=0.99\columnwidth]{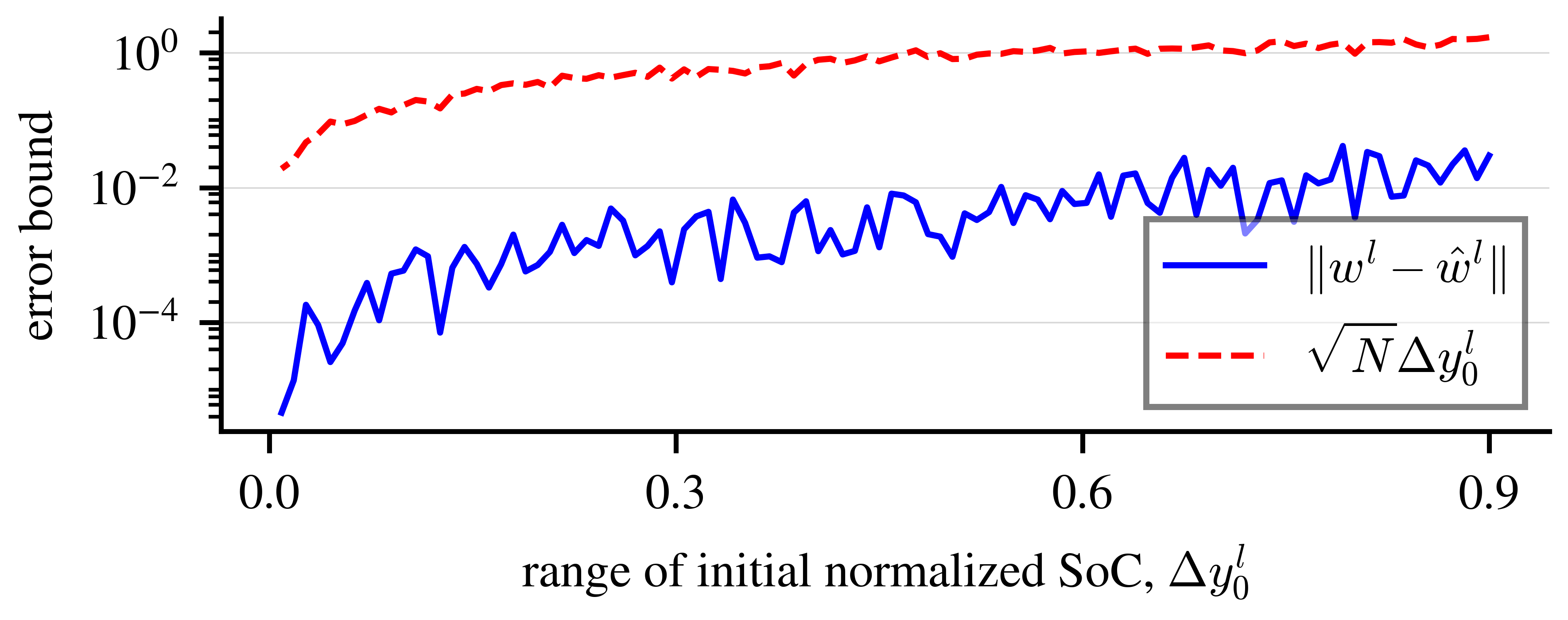}
    \caption{Verification of the LoMPC solution error bound in \eqref{eq:ev-w-error-bound} for large EVs with $M = 20$.
    For a given $\Delta y^l_0$, we randomly generate an incentive $\lambda^l$, compute the team-optimal solution $\hat{w}^l$ as the optimal solution of the EV LoMPC with $\theta^i = 0$ (see \eqref{eq:bounded-incentive-controllability-w-hat}), and compute $w^l$ using \eqref{eq:ev-w-error-bound}.
    The plot shows that for any desired average electricity consumption $\hat{w}^l$ (that is feasible), we can set the unit price of electricity $\lambda^l$ such that the actual average electricity consumption $w^l$ has a bounded error compared to $\hat{w}^l$.
    The error bound is given by $\sqrt{\horizon} \Delta y^l_0$, where $\Delta y^l_0$ is the range of normalized SoC of the large EV population.
    }
    \label{fig:ev-robustness-bound}
\end{figure}

In this subsection, we show that Assumptions~\ref{assum:properties-of-gi}, \ref{assum:linear-convex-incentive-controllability}, and \ref{assum:multiple-lompcs-properties-of-gi} hold for the EV charging example, and thus we can apply all the results derived in the paper.
In particular, we compute the error bound provided by Lem.~\ref{lem:bounded-incentive-controllability}.

Since the EV LoMPC costs are convex, Assumption~\ref{assum:properties-of-gi} is satisfied.
Moreover, the small EV LoMPC cost functions have a strong-convexity modulus of $\constsc^s = 2 \delta (\Theta^s)^2$, while those of large EVs have a strong-convexity modulus of $\constsc^l = 2 \delta (\Theta^l)^2$.

Next, we show that Assumption~\ref{assum:linear-convex-incentive-controllability} holds for both small EV and large EV incentives.
If $\phi^s(w^1) \preceq_{\mathset{K}} \phi^s(w^2)$, then from \eqref{eq:ev-phi-function}, $w^1 \leq w^2$ and $w^2 \leq w^1$ $\Rightarrow \ w^1 = w^2$.
For any $w^* \in \real^\horizon$,
\begin{equation*}
    D\phi^s(w^*)^\top = \Theta^s [I, \ -I, \ 2 \, \text{diag}(w^*)].
\end{equation*}
Then, for any $v \in \real^\horizon$, $D\phi^s(w^*)^\top\lambda = v$, where
\begin{equation*}
\lambda = (\max\{v, \mathbb{0}\}, -\min\{v, \mathbb{0}\}, \mathbb{0}) / \Theta^s \succeq_{\mathset{K}^*} 0.
\end{equation*}
Thus $\phi^s$, and similarly $\phi^l$, satisfies Assumption~\ref{assum:linear-convex-incentive-controllability}.

Finally, as discussed in Sec.~\ref{subsec:ev-lompc-formulation}, Assumption~\ref{assum:multiple-lompcs-properties-of-gi} holds for small EVs (and likewise for large EVs), with
\begin{equation}
\label{eq:ev-theta-bar}
\begin{gathered}
    \lVert \theta^i(\bm{y}_0)\rVert \leq \bar{\theta}^s = 2 \delta (\Theta^s)^2 \Delta y_0^s \sqrt{\horizon}, \\
    \Delta y_0^s = \frac{1}{2} \biggl(\max_{i \in \text{small EVs}} \{y^i_0\} - \min_{i \in \text{small EVs}} \{y^i_0\}\biggr).
\end{gathered}
\end{equation}

Since Assumptions~\ref{assum:linear-convex-incentive-controllability} and \ref{assum:multiple-lompcs-properties-of-gi} hold, we can use Lem.~\ref{lem:bounded-incentive-controllability} to provide an error bound on the small EV LoMPC solution.
Let $\hat{w}^s$ be such that $\mathbb{0} \leq \hat{w}^s \leq w^s_\text{max} \mathbb{1}$.
Then, we can find an incentive $\lambda^s \succeq_{\mathset{K}^*} 0$ such that
\begin{equation}
\label{eq:ev-w-error-bound}
\begin{gathered}
    \lVert \hat{w}^s - w^s \rVert \leq \bar{\theta}^s/\constsc^s = \sqrt{N} \Delta y_0^s, \\
    w^s = \frac{1}{M} \sum_{i \in \text{small EVs}} w^{i*}(\xi_0, \lambda^s),
\end{gathered}
\end{equation}
where $w^{i*}(\xi_0, \lambda)$ is the optimal solution of the EV $i$ LoMPC.
In particular, we note that the error bound in \eqref{eq:ev-w-error-bound} depends on the horizon $N$ and the range of initial normalized SoCs $\Delta y^s_0$, but not on the population size $M$.
We verify the dependence of the error bound on $\Delta y^s_0$ in Fig.~\ref{fig:ev-robustness-bound}.

For a given $\hat{w}^s$ such that $\mathbb{0} \leq \hat{w}^s \leq w^s_\text{max} \mathbb{1}$, an optimal incentive $\lambda^{s*}$ is computed using the incentive solver given by Alg.~\ref{alg:iterative-method}.
In Fig.~\ref{fig:ev-dual-cost-decrease}, we verify the incentive solver, derived from Thm.~\ref{thm:linear-convex-incentives-iterative-method}, by checking the dual cost decrease condition \eqref{eq:linear-convex-incentives-iterative-method-dual-cost-decrease} for each solver iteration.

\begin{figure}
    \centering
    \includegraphics[width=0.99\columnwidth]{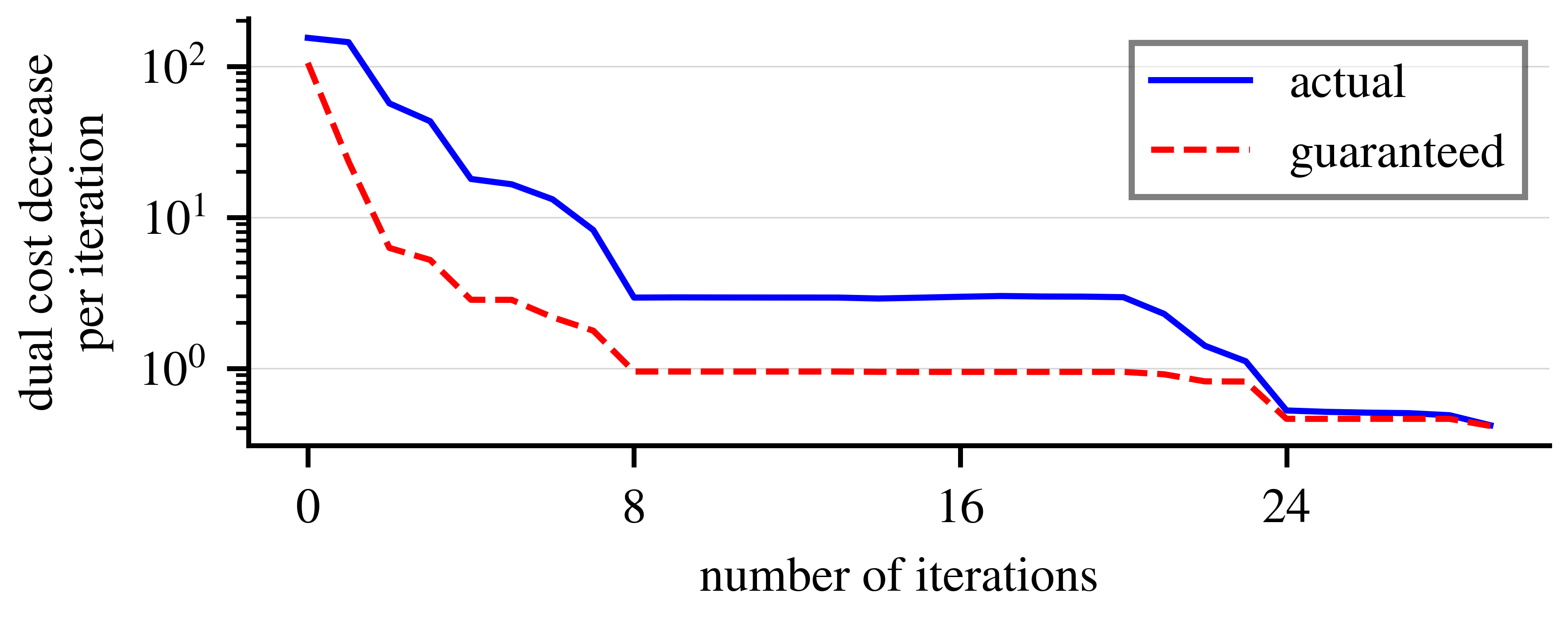}
    \caption{Verification of the incentive solver in Alg.~\ref{alg:iterative-method}.
    The incentive solver uses the majorization-minimization method to compute an optimal incentive $\lambda^*$ iteratively.
    For each iterate $\lambda^{(k)}$, the dual cost $-\tilde{g}^*$ is majorized by $-\tilde{g}^*(\cdot \,; \lambda^{(k)})$, which is then minimized (see Thm.~\ref{thm:linear-convex-incentives-iterative-method}).
    The solid line shows the actual decrease in the dual cost function $-\tilde{g}^*$ at each iteration, which is guaranteed to be greater than the decrease in $-\tilde{g}^*(\cdot \,; \lambda^{(k)})$ (the dashed line).
    The plot is shown for large EVs with $M = 200$.
    }
    \label{fig:ev-dual-cost-decrease}
\end{figure}

For the robust BiMPC \eqref{eq:robust-bimpc}, we use a robustness horizon $\horizon_r = 1$, i.e. we only account for the error in LoMPC solution at time step $0$ (see Rem.~\ref{rem:robustness-horizon}).
In Sec.~\ref{subsec:ev-bimpc-formulation}, we find conditions under which the robust BiMPC \eqref{eq:robust-bimpc} with $\horizon_r = 1$ is always feasible.
The LoMPC solution error bound at time step $0$ can be found from \eqref{eq:ev-w-error-bound} as follows:
\begin{equation}
\label{eq:ev-w0-error-bound}
    | \hat{w}^s_0 - w^s_0 | \leq \sqrt{N} \Delta y_0^s, \quad | \hat{w}^l_0 - w^l_0 | \leq \sqrt{N} \Delta y_0^l.
\end{equation}

\subsection{BiMPC Formulation}
\label{subsec:ev-bimpc-formulation}

The ISO decides the amount of electricity to be generated in each time interval, depending on the SoCs of the EV population and a forecast of the external electricity demand~\cite{zou2017efficient}.
In addition to the model in \cite{zou2017efficient}, we assume that the ISO controls a storage battery that can be charged or depleted depending on the total electricity demand.
The ISO has state variable $x_k \in \real$, which denotes the state of charge of its storage battery at the time step $k$.
The amount of electricity generated in time interval $k$ is denoted by $u^g_k \in \real$, and the electricity discharged from the storage battery in time interval $k$ is denoted by $u^b_k \in \real$.
The unit prices of electricity for small and large EVs are given by $\lambda^s, \lambda^l \in \mathset{K}^*$, respectively.
The state dynamics are as follows:
\begin{equation}
\label{eq:ev-bimpc-dynamics}
\begin{gathered}
    x_{k+1} = x_k + u^b_k, \\
    u^g_k = (M/2) (\Theta^s w^s_k + \Theta^l w^l_k) + D_k + u^b_k, \\
    0 \leq x_k \leq x_\text{max}, \\
    0 \leq u^g_k \leq u^g_\text{max}, \quad -u^b_\text{max} \leq u^b_k \leq u^b_\text{max},
\end{gathered}
\end{equation}
where $w^s$ (and likewise $w^l$), defined in \eqref{eq:ev-w-error-bound}, is the average electricity consumption for small EVs corresponding to unit price $\lambda^s$.
$D_k$ is the forecast of external electricity demand in time interval $k$~\cite{zou2017efficient}, and $x_\text{max}$ is the capacity of the storage battery.
$u^g_\text{max}$ is the maximum electricity generation in a time interval, and $u^b_\text{max}$ is the maximum charge/discharge amount for the storage battery in a time interval.

The BiMPC cost comprises electricity generation cost and aggregate charging cost~\cite{zou2017efficient}.
The electricity generation cost penalizes high electricity generation and is given by
\begin{equation}
\label{eq:ev-bimpc-generation-cost}
    \textstyle{\sum}_{k=0}^{\horizon-1} \, c^g (u^g_k)^{1.7},
\end{equation}
where $c^g > 0$ is a parameter.
The aggregate charging cost is similar to the LoMPC charging cost and is given by
\begin{align}
\label{eq:ev-bimpc-tracking-cost}
    \smash{\sum_{k=1}^N} \, \gamma^{k-N} \bigl[ \, & ( \textstyle{\sum}_{k'=1}^k w^s_{k'} - (y^s_\text{max} - y^s_{0, m}))^2 \\
    & + ( \textstyle{\sum}_{k'=1}^k w^l_{k'} - (y^l_\text{max} - y^l_{0, m}))^2 \,\bigr], \nonumber
\end{align}
where $y^s_{0, m}$ and $y^l_{0, m}$ are mean initial normalized SoCs and $\gamma > 1$ is a parameter.
For the robust BiMPC formulation, we define $\hat{w}^s$ (and likewise $\hat{w}^l$) as the desired average electricity consumption for small EVs.
Then, for any $\hat{w}^s$ satisfying $\mathbb{0} \leq \hat{w}^s \leq w^s_\text{max} \mathbb{1}$, we can find an incentive $\lambda^s$ such that the error bounds in \eqref{eq:ev-w-error-bound} and \eqref{eq:ev-w0-error-bound} hold.
Using the LoMPC solution error bound \eqref{eq:ev-w0-error-bound} for robustness horizon $\horizon_r = 1$, we can write the constraints of the robust BiMPC \eqref{eq:robust-bimpc} as (see Rem.~\ref{rem:explicit-robust-bimpc-formulation}) follows:
\begin{equation}
\label{eq:ev-robust-bimpc-constraints}
\begin{gathered}
    \hat{u}^b = u^g - D - (M/2)(\Theta^s \hat{w}^s + \Theta^l \hat{w}^l), \\
    x_{k+1} = x_k + \hat{u}^b_k, \quad k \in \langle N\rangle, \\
    \mathbb{0} \leq u^g \leq u^g_\text{max} \mathbb{1}, \\
    -u^b_\text{max}\mathbb{1} + \Delta e_1 \leq \hat{u}^b \leq u^b_\text{max}\mathbb{1} - \Delta e_1, \\
    \Delta \mathbb{1} \leq \hat{x}_{1:\horizon} \leq x_\text{max} - \Delta \mathbb{1},
\end{gathered}
\end{equation}
where $e_1 = (1, 0, ..., 0)$ is a unit vector, and
\begin{equation}
\label{eq:ev-robust-bimpc-delta}
    \Delta = (M/2) \sqrt{N} (\Theta^s \Delta y^s_0 + \Theta^l \Delta y^l_0).
\end{equation}
Note that the constraint bounds in \eqref{eq:ev-robust-bimpc-constraints}, as compared to \eqref{eq:ev-bimpc-dynamics}, are reduced by an amount proportional to $\Delta$.
Moreover, if $u^g_\text{max}$ is large enough, the set of constraints \eqref{eq:ev-robust-bimpc-constraints} is feasible if, $\Delta \leq x_0 \leq x_\text{max} - \Delta$ and
\begin{equation}
\label{eq:ev-robust-constraints-feasibility}
    \Delta \leq u^b_\text{max}, \quad 2\Delta \leq x_\text{max}.
\end{equation}
We can find a feasible vector by setting $u^g$ such that $\hat{u}^b = \mathbb{0}$.

The Robust BiMPC problem \eqref{eq:robust-bimpc} is as follows:
\begin{align}
\label{eq:ev-bimpc}
    (u^{g*}, \hat{w}^{s*}, \hat{w}^{l*}) (\xi_0) = \argmin_{u^g, \hat{w}^s, \hat{w}^l} \ & f(u, \hat{w}^s, \hat{w}^l; \xi_0), \\
    \text{s.t} \quad &
    (u^g, \hat{w}^s, \hat{w}^l) \text{ satisfies \eqref{eq:ev-robust-bimpc-constraints}}, \nonumber
\end{align}
where the cost function $f$ is the sum of \eqref{eq:ev-bimpc-generation-cost} and \eqref{eq:ev-bimpc-tracking-cost} with $(w^s, w^l) = (\hat{w}^s, \hat{w}^l)$.
In the hierarchical MPC formulation (see Sec.~\ref{subsec:hierarchical-formulation}), we don't impose convexity assumptions on the BiMPC cost $f$.
Thus, more complex costs can be incorporated into the BiMPC formulation without affecting the solution.

Lastly, we describe the quantities known to the ISO.
The ISO knows the values of $N$, $\Theta^s$, $\Theta^l$, $w^s_\text{max}$, and $w^l_\text{max}$.
Notably, the battery degradation costs are private to the EVs.
Additionally, at each time step, the ISO gets the current normalized SoCs of all EVs and thus knows $\xi_0 = (x_0, \bm{y}_0)$.
These are the same as the assumptions in \cite{zou2017efficient}.
We assume that EVs are truthful when reporting their SoCs.

\begin{remark} (Differential pricing based on initial SoC)
\label{rem:ev-differential-pricing}
For smaller values of $\Delta$, the feasible set of \eqref{eq:ev-bimpc} is larger.
To reduce the value of $\Delta$, we can reduce $\Delta y^s_0$ and $\Delta y^l_0$ by differentially pricing EVs depending on their initial normalized SoC $y^i_0$.
As an example, the unit price of electricity $\lambda^s$ for a small EV $i$ with initial SoC $y^i_0 \in [0.3, 0.4)$ might be different from that of a small EV $j$ with $y^j_0 \in [0.4, 0.5)$.
Differential pricing allows for tighter control on $\Delta y^s_0, \Delta y^l_0$ for EVs in each interval, and thus the error bound $\Delta$.
In the extreme case, we set different prices for each EV, allowing us to get an error bound of $0$.
In general, we can partition the set of small and large EVs into $P$ sets depending on $y^i_0$ and set electricity prices accordingly (see Rem.~\ref{rem:partition-of-lompcs}).
For clarity of presentation, we don't consider differential pricing in the BiMPC formulation \eqref{eq:ev-bimpc}, but we do so in the implementation\footref{code}.
\end{remark}

\subsection{Simulation Setup}
\label{subsec:ev-simulation-setup}

\begin{figure}
    \centering
    \includegraphics[width=0.99\columnwidth]{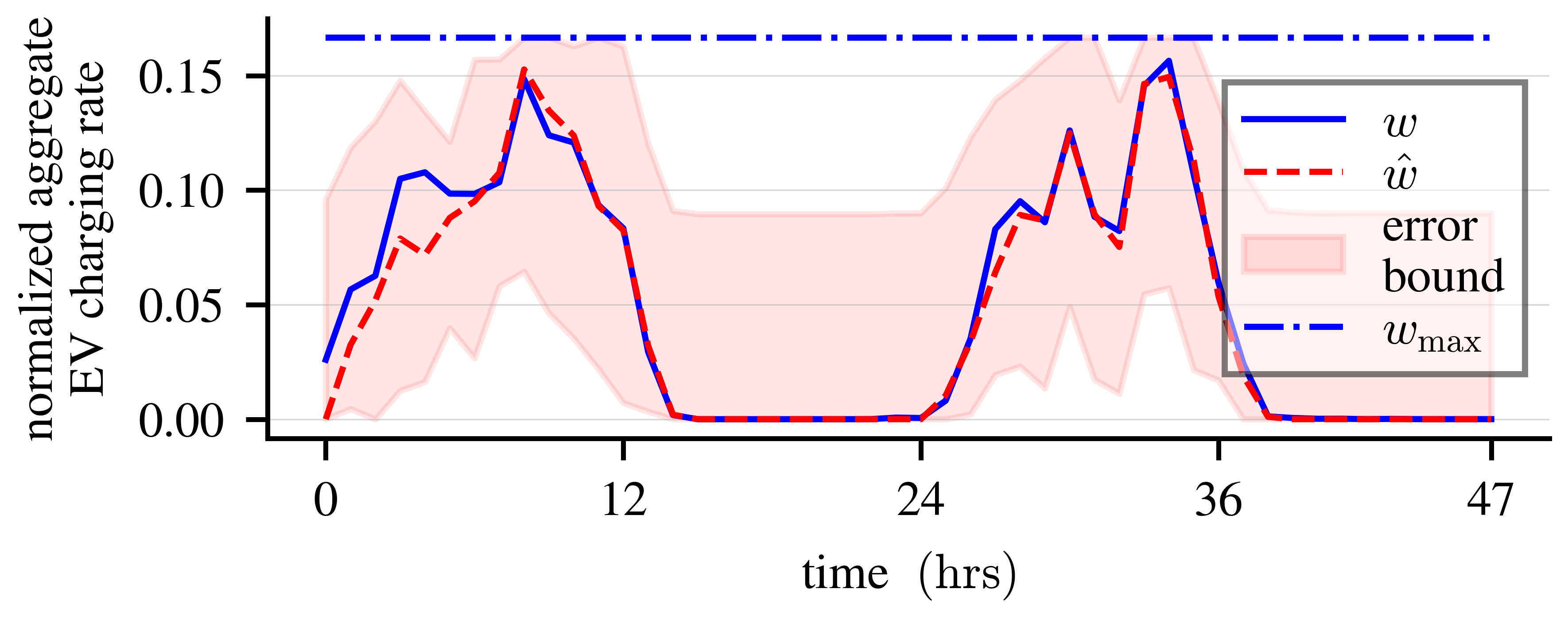}
    \caption{The total electricity consumption by the EVs over $48$ hrs, normalized by $B$ (see \eqref{eq:ev-lompc-b}).
    The solid line shows the actual electricity consumption $w$, while the dotted line shows the predicted electricity consumption $\hat{w}$ (the team-optimal solution from the BiMPC).
    The colored region shows the error bound; the actual consumption $w$ is guaranteed to be within the error bound range around the predicted consumption $\hat{w}$.
    The actual consumption reduces during peak external demand and rises when external demand is low.
    }
    \label{fig:ev-aggregate-charging-rate}
\end{figure}

The external electricity demand vector $D$, shown in Fig.~\ref{fig:ev-demand-electricity-generation}, is obtained for a day in August 2024 in the Midcontinent ISO region of North America and scaled down by a factor of $4000$.
The external demand is at its peak during the afternoon and at its lowest during the early morning.
The simulations are shown for a duration of $48$ hrs, with MPC horizon length $\horizon = 12$ hrs and an EV population size of $1000$.
All the results are normalized with respect to the total EV battery capacity $B$, as defined in \eqref{eq:ev-lompc-b}.

All the convex optimization problems (the robust BiMPC in \eqref{eq:ev-bimpc}, the incentive solver $\lambda$ update step in line \ref{alg-line:lambda-update} Alg.~\ref{alg:iterative-method}, and the EV LoMPCs) are formulated using CVXPY \cite{diamond2016cvxpy} in parametric form and solved using the CLARABEL solver \cite{goulart2024clarabel}.
The simulation parameters and the Python code for all the results can be found in the repository\footref{code}.

\subsection{Results}
\label{subsec:ev-results}

\begin{figure}
    \centering
    \includegraphics[width=0.99\columnwidth]{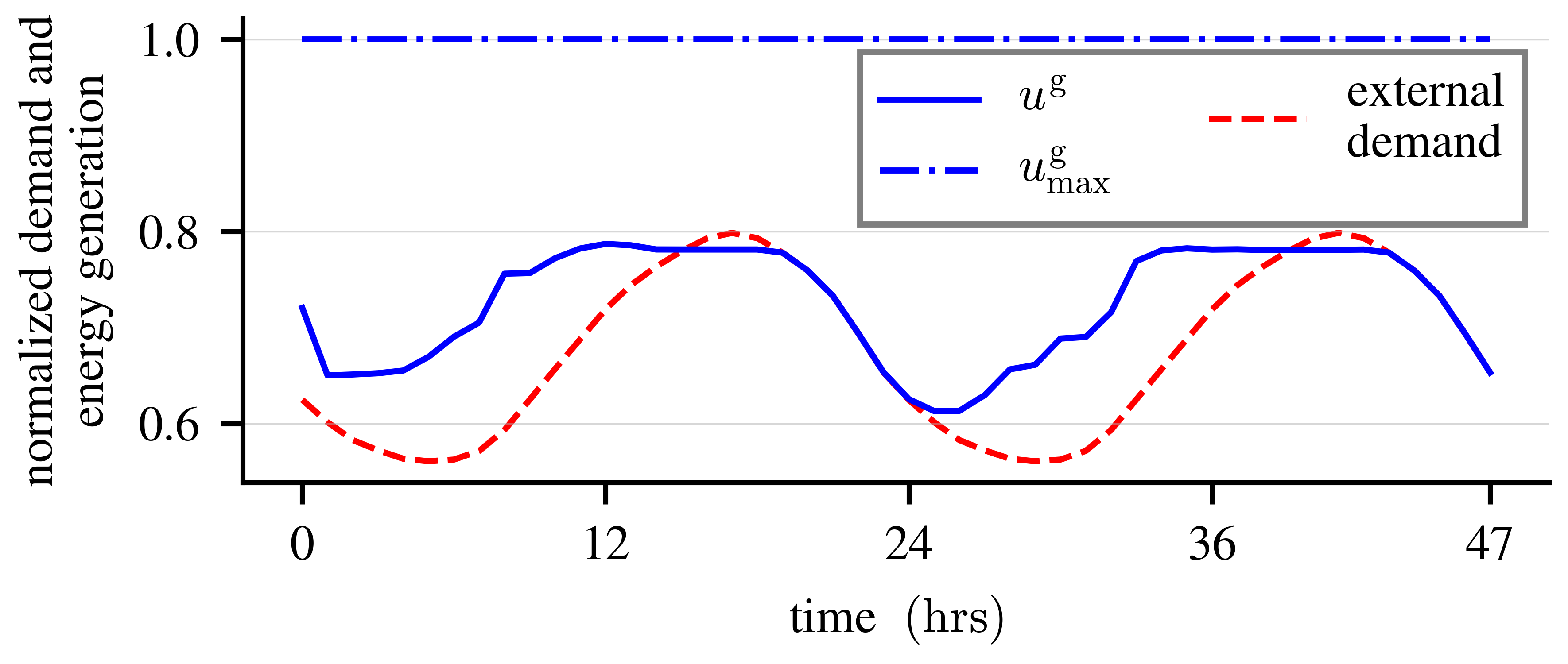}
    \caption{Electricity generation by the ISO over $48$ hrs, normalized by $B$ (see \eqref{eq:ev-lompc-b}).
    The solid line shows the amount of electricity generated $u^g$, while the dotted line shows the external demand.
    The total electricity supply from the ISO always fulfills the total demand.
    The unit price of electricity is set by the ISO so that the total EV electricity consumption fills the overnight demand valley~\cite{ma2013decentralized,zou2017efficient}.
    }
    \label{fig:ev-demand-electricity-generation}
\end{figure}

\begin{figure}
    \centering
    \includegraphics[width = 0.99\linewidth]{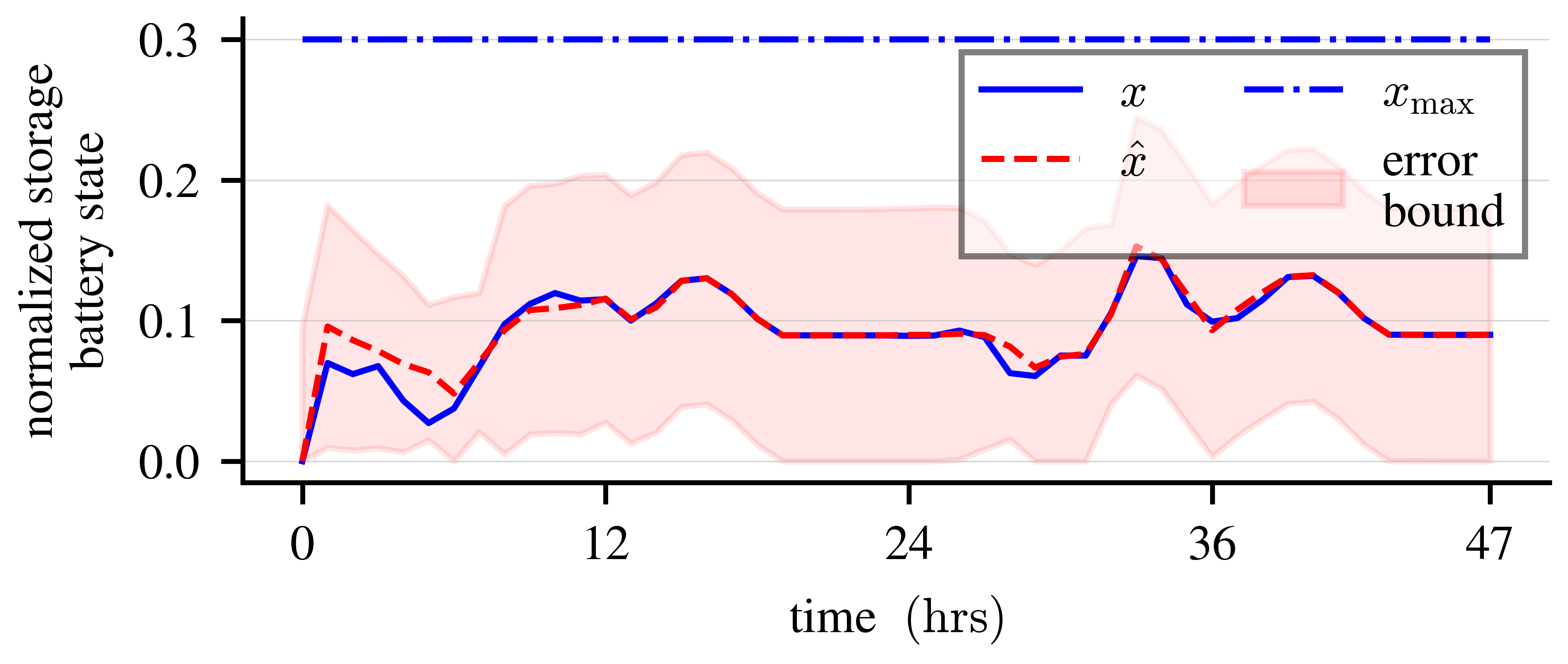}
    \caption{The SoC of the storage battery over $48$ hrs, normalized by $B$ (see \eqref{eq:ev-lompc-b}).
    The solid line shows the actual SoC, while the dotted line shows the predicted SoC (from the BiMPC).
    The colored region shows the error bound; the actual SoC is guaranteed to lie within the error bound range of the predicted SoC.
    We note that the plot only shows the usable SoC range of the storage battery; most batteries have a minimum SoC, which is used as a baseline for the plot.
    The ISO stores the excess electricity generated in a storage battery.
    The storage battery is charged during low external demand and depleted during peak external demand to reduce electricity generation costs.
    Since the total EV electricity consumption is uncertain, the ISO always maintains a nonzero storage battery SoC.
    }
    \label{fig:ev-storage-battery-state}
\end{figure}

The solution trajectories obtained from the simulation are shown in Figs.~\ref{fig:ev-aggregate-charging-rate}, \ref{fig:ev-demand-electricity-generation}, and \ref{fig:ev-storage-battery-state}.
Fig.~\ref{fig:ev-demand-electricity-generation} shows that the ISO sets the unit price of electricity to ensure that electricity generation is evenly spread throughout the day.
Fig.~\ref{fig:ev-aggregate-charging-rate} shows the total electricity consumption by the EV population.
We can observe that EV electricity consumption is close to zero during the afternoon and high during the early morning.
Fig.~\ref{fig:ev-storage-battery-state} shows the state of charge of the storage battery.
The ISO charges the storage battery during low external demand and depletes it during peak external demand.
Moreover, the ISO always maintains a nonzero storage battery SoC to ensure that unexpected EV electricity demands can be satisfied.

During the simulation, a total of $3793$ small EVs and $1715$ large EVs are fully charged, which is $38.7\%$ and $29.2\%$ of the maximum possible value (at the highest charging rate), respectively.
On average, the incentive solver requires $14.7$ iterations and $22.7$ iterations for convergence for small and large EVs, respectively.

Ideally, the total electricity consumption by the EVs is such that there are no spikes in electricity generation.
This type of electricity consumption is called valley-filling~\cite{ma2013decentralized,zou2017efficient}.
From Figs.~\ref{fig:ev-aggregate-charging-rate} and \ref{fig:ev-demand-electricity-generation}, we can see that spikes in electricity generation are avoided by shifting EV consumption loads towards periods of low external demand.
In our formulation, we allow EVs to set their optimal charging rates while also ensuring that the goals of the ISO are also met.

\section{Conclusions}
\label{sec:conclusions}

In this paper, we presented a control framework for incentive hierarchical MPCs with multiple LoMPCs.
We defined team-optimal solutions and incentive controllability for the incentive hierarchical MPC.
We defined the dual of the incentive LoMPC cost function and, using tools from convex analysis, established a connection between the minima of the dual function and the set of optimal incentives.
We showed that, under some assumptions, the hierarchical MPC problem with a linear-convex incentive structure is incentive controllable, and the optimal incentives can be computed iteratively.
We also proved that, for the case of multiple LoMPCs, the team-optimal solution can be attained up to an error bound and provided an algorithm for computing optimal incentives.
Notably, the proposed algorithm scales well with the number of LoMPCs and does not require knowing the LoMPC cost function.
Finally, we considered dynamic electricity price control for large-scale EV charging.
We used our algorithm to show that an ISO can set the electricity prices to approximately minimize a social cost while also guaranteeing constraint satisfaction.
Our BiMPC formulation resulted in an electricity consumption pattern that fills the overnight demand valley.

\appendices

\section{Proofs}
\label{app:proofs}

\subsection{Proof of Lem.~\ref{lem:linear-convex-incentive-conjugate-grad}}
\label{subapp:proof-linear-convex-incentive-conjugate-grad}

Since $\langle \lambda, \phi(w)\rangle$ is linear in $\lambda$ for each $w$ and $\bar{g}^*$ is the pointwise minimum of a family of affine functions in $\lambda$, $-\bar{g}^*$ is convex~\cite[Sec.~3.2.3]{boyd2004convex}.
Since, for each $\lambda \succeq_{\mathset{K}^*} 0$, $g + \langle \lambda, \phi(\cdot)\rangle$ is $\constsc$-strongly convex (by Assumption~\ref{assum:properties-of-gi}) with a closed convex domain, \eqref{eq:linear-convex-dual} has a unique minimum and $-\bar{g}^*(\lambda) < \infty$ for each $\lambda \in \mathset{K}^*$.
So, $\dom{\bar{g}^*} = \mathset{K}^*$.

By Danskin's Theorem \cite[Prop.~B.25]{bertsekas1997nonlinear}%
, $-\bar{g}^*$ is subdifferentiable on $\setint(\mathset{K}^*)$ with subdifferential given by
\begin{equation*}
    \partial (-\bar{g}^*)(\lambda) = \bigl\{ -\phi(\bar{w}): \bar{w} \in \textstyle{\argmin}_w \{g(w) + \langle \lambda, \phi(w)\rangle\} \bigr\}.
\end{equation*}
Using the fact that \eqref{eq:linear-convex-dual} has a unique minimum for all $\lambda \in \mathset{K}^*$, we get that $\bar{g}^*$ is differentiable on $\setint(\mathset{K}^*)$ and that its gradient is given by \eqref{eq:linear-convex-incentive-conjugate-grad}.
If $\lambda \in \setbd(\mathset{K}^*)$ (the boundary of $\mathset{K}^*$), $-\phi(w^*(\lambda))$ is a subgradient of $-\bar{g}^*$.
For ease of notation, we assign $\nabla \bar{g}^*(\lambda) = \phi(w^*(\lambda))$ when $\lambda \in \setbd(\mathset{K}^*)$.

Next, we prove the continuity of $w^*(\cdot)$.
Let $\{\lambda^{(k)} \} \rightarrow \bar{\lambda}$, $\lambda^{(k)} \in \mathset{K}^*$, $\lVert \lambda^{(k)} \rVert \leq M$, and $w^{(k)} = w^*(\lambda^{(k)})$.
First, we show that $\{w^{(k)}\}$ is bounded.
Choose any $w' \in \dom{g}$ and $v \in \partial g(w')$.
Then, by $\constsc$-strong convexity of $g$ and $\mathset{K}$-convexity of $\phi$ (see Sec.~\ref{subsubsec:strong-convexity-and-lipschitz-smoothness}),
\begin{align*}
    g(w) + \langle \lambda^{(k)}, \phi(w)\rangle & \geq g(w') + \langle v, w - w'\rangle + \frac{m}{2} \lVert w - w'\rVert^2 \\
    & \text{\phantom{$\geq$}} + \langle \lambda^{(k)}, \phi(w') + D\phi(w') [w - w']\rangle, \\
    & \geq g(w') - M \lVert \phi(w')\rVert + \frac{m}{2} \lVert w - w'\rVert^2 \\
    & \text{\phantom{$\geq$}} - (\lVert v\rVert + M \lVert D \phi(w') \rVert) \lVert w - w'\rVert.
\end{align*}
Since $g(w') + \langle \lambda^{(k)}, \phi(w')\rangle \leq g(w') + M \lVert \phi(w')\rVert$, we have
\begin{equation*}
    \lVert w^{(k)} - w' \rVert \leq \frac{ \lVert v\rVert + M \lVert D \phi(w') \rVert + \sqrt{\constsc M \lVert \phi(w')\rVert}}{(\constsc/2)}.
\end{equation*}
Thus $\{w^{(k)} \}$ is bounded.
Let $\bar{w}$ be a limit point of $\{w^{(k)}\}$ and $\{w^{(l)}\}$ be the corresponding subsequence.
Then,
\begin{align*}
    \bar{g}^*(\bar{\lambda}) & \leq g(\bar{w}) + \langle \bar{\lambda}, \phi(\bar{w})\rangle = \lim_{l \rightarrow \infty} \bigl(g(w^{(l)}) + \langle \lambda^{(l)}, \phi(w^{(l)})\rangle \bigr), \\
    & = \lim_{l \rightarrow \infty} \bar{g}^*(\lambda^{(l)}) \leq \lim_{l \rightarrow \infty} \bigl(g(w^*(\bar{\lambda})) + \langle \lambda^{(l)}, \phi(w^*(\bar{\lambda}))\rangle\bigr), \\
    & = g(w^*(\bar{\lambda})) + \langle \bar{\lambda}, \phi(w^*(\bar{\lambda}))\rangle = \bar{g}^*(\bar{\lambda}).
\end{align*}
So, any limit point $\bar{w}$ of $\{w^{(k)}\}$ minimizes $g(w) + \langle \bar{\lambda}, \phi(w)\rangle$.
Therefore, there is a unique limit point of $\{w^{(k)}\}$ and $\{w^{(k)}\} \rightarrow w^*(\bar{\lambda})$, i.e. $w^*(\cdot)$ is continuous.
\hidetext{Alternatively, we can use the Maximum Theorem~\cite[Sec.~E.3]{ok2007real}.}

\subsection{Proof of Lem.~\ref{lem:linear-convex-lipschitz-smooth-dual}}
\label{subapp:proof-linear-convex-lipschitz-smooth-dual}

From Lem.~\ref{lem:linear-convex-incentive-conjugate-grad}, $-\bar{g}^*$ is a convex function.
Since $\bar{w} = w^*(\bar{\lambda})$, the following optimality condition must hold~\cite[Thm.~23.8]{rockafellar1997convex}:
\begin{equation*}
    0 \in \partial g(\bar{w}) + \{D \phi(\bar{w})^\top \bar{\lambda}\}.
\end{equation*}
So, $\exists v \in \partial g(\bar{w})$ such that $v = -D \phi(\bar{w})^\top \bar{\lambda}$.
$g$ is $\constsc$-strongly convex under Assumption~\ref{assum:properties-of-gi}, and so for all $w \in \dom{g}$ we have that (see Sec.~\ref{subsubsec:strong-convexity-and-lipschitz-smoothness}
)
\begin{equation*}
    g(w) \geq g(\bar{w}) + \langle v, w - \bar{w}\rangle + (\constsc/2) \lVert w - \bar{w}\rVert^2.
\end{equation*}
Under Def.~\ref{def:linear-convex-incentive}, for any $\lambda \in \mathset{K}^*$, $\langle \lambda, \phi(\cdot) \rangle$ is convex and differentiable.
By \eqref{eq:subdifferential-definition}, we have that for all $w \in \dom{g}$,
\begin{equation*}
\begin{split}
    \langle \lambda, \phi(w)\rangle & \geq \langle \lambda, \phi(\bar{w})\rangle + \langle \lambda, D \phi(\bar{w})[w - \bar{w}]\rangle, \\
    & = \langle \lambda, \phi(\bar{w})\rangle + \langle D \phi(\bar{w})^\top \lambda, w - \bar{w}\rangle.
\end{split}
\end{equation*}
Adding the above two inequalities and using the fact that $v = -D \phi(\bar{w})^\top \bar{\lambda}$, we get that for all $w \in \dom{g}$ and $\lambda \in \mathset{K}^*$,
\begin{equation*}
\begin{split}
    g(w) + \langle \lambda, \phi(w)\rangle & \geq g(\bar{w}) + \langle \lambda, \phi(\bar{w})\rangle + (\constsc/2) \lVert w - \bar{w}\rVert^2\\
    & \text{\phantom{$\geq$}} + \langle D \phi(\bar{w})^\top (\lambda - \bar{\lambda}), w - \bar{w}\rangle.
\end{split}
\end{equation*}
Note that for a fixed $\lambda$, the RHS of the above inequality is quadratic in $w$.
Minimizing both the LHS and the RHS over $w \in \real^{\Nw\horizon}$ and using \eqref{eq:linear-convex-dual}, we get that for all $\lambda \in \mathset{K}^*$, 
\begin{equation*}
\begin{split}
    \bar{g}^*(\lambda) & \geq g(\bar{w}) + \langle \lambda, \phi(\bar{w})\rangle - \lVert D \phi(\bar{w})^\top (\lambda - \bar{\lambda}) \rVert^2 / (2\constsc), \\
    & = \bar{g}^*(\bar{\lambda}) + \langle \lambda - \bar{\lambda}, \phi(\bar{w})\rangle - \lVert D \phi(\bar{w})^\top (\lambda - \bar{\lambda}) \rVert^2 / (2\constsc), \!\!
\end{split}
\end{equation*}
which is the same as \eqref{eq:linear-convex-lipschitz-smooth}.

Let $\Lambda \subset \mathset{K}^*$ be any compact set.
We show that the restriction $-\bar{g}^*|_{\Lambda}$ is Lipschitz smooth.
By continuity of the function $w^*(\cdot)$ defined in \eqref{eq:linear-convex-incentive-conjugate-grad-w} (Lem.~\ref{lem:linear-convex-incentive-conjugate-grad}), $w^*(\Lambda)$ is a compact set.
By continuous differentiability of $\phi$ (Def.~\ref{def:linear-convex-incentive}), $\lVert D\phi(w)\rVert$ is bounded on $w^*(\Lambda)$.
Then, Lipschitz smoothness of $-\bar{g}^*|_{\Lambda}$ follows from \eqref{eq:linear-convex-lipschitz-smooth} and \cite[Thm.~2.1.5]{nesterov2018lectures}.

\subsection{Proof of Thm.~\ref{thm:linear-convex-incentives-iterative-method}}
\label{subapp:proof-linear-convex-incentives-iterative-method}

From Lem.~\ref{lem:linear-convex-incentive-conjugate-grad}, $-\tilde{g}^*$ is a differentiable convex function on $\mathset{K}^*$.
The gradient of $\tilde{g}^*$ can be computed using \eqref{eq:linear-convex-incentive-conjugate-grad} as follows:
\begin{equation*}
    \nabla (-\tilde{g}^*)(\lambda) = -\phi(w^*(\lambda)) + \phi(w^*).
\end{equation*}
By the incentive controllability assumption, Assum.~\ref{assum:linear-convex-incentive-controllability}, $\exists \lambda^* \in \mathset{K}^*$ such that $w^*(\lambda^*) = w^*$.
So, $\nabla (-\tilde{g}^*)(\lambda^*) = 0$, i.e., $\lambda^*$ is a minimizer of $-\tilde{g}^*$, and $-\tilde{g}^*$ is lower bounded.
Next, using \eqref{eq:linear-convex-lipschitz-smooth} with $\bar{\lambda} = \lambda^{(k)}$ and $\bar{w} = w^{(k)}$, we get that
\begin{align}
\label{eq:linear-convex-incentive-iterative-method-proof-majorization}
    -\tilde{g}^*(\lambda) & \leq -\tilde{g}^*(\lambda; \lambda^{(k)}), \\
    & \leq -\hat{g}^*(\lambda; \lambda^{(k)}) := -\tilde{g}^*(\lambda; \lambda^{(k)}) + \epsilon^{(k)} \lVert \lambda - \lambda^{(k)} \rVert^2, \nonumber
\end{align}
Thus, $-\hat{g}^*(\cdot \,; \lambda^{(k)})$ is a differentiable convex function that majorizes $-\tilde{g}^*(\lambda)$, with $-\hat{g}^*(\lambda^{(k)}; \lambda^{(k)}) = -\tilde{g}^*(\lambda^{(k)})$.
Using Lem.~\ref{lem:linear-convex-lipschitz-smooth-dual}, and the assumption that $\{\epsilon^{(k)}\}$ and $\{\lambda^{(k)} \}$ are bounded, $-\hat{g}^*(\cdot \,; \lambda^{(k)})$ is Lipschitz smooth, with a constant independent of $k$.
Thus, by the convergence of gradient descent for Lipschitz smooth functions, any limit point $\lambda^*$ of $\{\lambda^{(k)} \}$ is a stationary point of $-\tilde{g}^*$ \cite[Prop.~2.3.2]{bertsekas1997nonlinear}\hidetext{$\ $(also see \cite[Sec.~5]{pilanci2018ee364b})}.
Because $\dom(-\tilde{g}^*) = \mathset{K}^*$, by the necessary and sufficient condition for optimality of convex problems \cite[Prop.~2.1.2(a)]{bertsekas1997nonlinear},
\begin{align*}
    & \nabla (-\tilde{g}^*)(\lambda^*) = \phi(w^*) - \phi(w^*(\lambda^*)) \in \mathset{K}, \\
    & \Rightarrow \phi(w^*) \succeq_{\mathset{K}} \phi(w^*(\lambda^*)).
\end{align*}
By Assumption~\ref{assum:linear-convex-incentive-controllability}, this implies that $w^* = w^*(\lambda^*)$ for all limit points $\lambda^*$ of $\{ \lambda^{(k)}\}$.
By continuity of $w^*(\cdot)$ (Lem.~\ref{lem:linear-convex-incentive-conjugate-grad}), $\lim_{k \rightarrow \infty} w^{(k)} = \lim_{k \rightarrow \infty} w^*(\lambda^{(k)}) = w^*(\lambda^*) = w^*$.

The dual cost decrease guarantee \eqref{eq:linear-convex-incentives-iterative-method-dual-cost-decrease} follows from \eqref{eq:linear-convex-incentive-iterative-method-proof-majorization} applied to $\lambda^{(k+1)}$ and the equality $-\tilde{g}^*(\lambda^{(k)}) = -\hat{g}^*(\lambda^{(k)}; \lambda^{(k)})$.

Let $\lambda^* \in \mathset{K}^*$ be an optimal incentive for $w^*$.
Then, by Lem.~\ref{lem:linear-convex-incentive-conjugate-grad}, $\nabla (-\tilde{g}^*)(\lambda^*) = 0$.
Using \eqref{eq:subdifferential-definition}, Lem.~\ref{lem:linear-convex-incentive-conjugate-grad}, and Lem.~\ref{lem:linear-convex-lipschitz-smooth-dual}, we have that for any $\lambda \in \mathset{K}^*$,
\begin{equation*}
-\tilde{g}^*(\lambda^*) \leq -\tilde{g}^*(\lambda) \leq -\tilde{g}^*(\lambda^*) + \frac{1}{2\constsc} \lVert D\phi(w^*)^\top (\lambda - \lambda^*) \rVert^2.
\end{equation*}
Thus, if $D\phi(w^*)^\top (\bar{\lambda} - \lambda^*) = 0$ for some $\bar{\lambda} \in \mathset{K}^*$, then $-\tilde{g}^*(\lambda^*) = -\tilde{g}^*(\bar{\lambda})$, i.e., $\bar{\lambda}$ minimizes $-\tilde{g}^*(\lambda)$.
Thus, $\bar{\lambda}$ is an optimal incentive.

\section{Robust BiMPC Formulation}
\label{app:robust-bimpc-formulation}

To write the robust BiMPC formulation \eqref{eq:robust-bimpc} for the problem structure defined in \eqref{eq:bimpc}, we can explicitly write the constraints \eqref{subeq:robust-bimpc-w-g-feasibility} and \eqref{subeq:robust-bimpc-w-f-feasibility}.
We assume, for now, that the state constraint sets $\mathset{X}$ and $\mathset{X}_\Omega$ in \eqref{eq:bimpc} are polytopic.
In other words, we assume that the state constraints $x_k \in \mathset{X}$ for all $k \in \langle \horizon \rangle$ and $x_\horizon \in \mathset{X}_\Omega$ of the BiMPC \eqref{eq:bimpc} are encoded by $\{x \in \real^{\Nx(\horizon+1)}: C x \leq d\}$.
%
Given $(\xi_0, u, \hat{w})$, the constraint \eqref{subeq:robust-bimpc-w-g-feasibility} can be expressed using the LoMPC problem \eqref{eq:ith-lompc} as $\hat{w}_k \in \mathset{W}, \ k \in \langle \horizon \rangle$ (recall from Rem.~\ref{rem:multiple-lompcs-no-state-constraints} that the LoMPCs cannot have any state constraints).
The domain of the BiMPC cost function, $\dom{f}$, in the BiMPC constraint \eqref{subeq:robust-bimpc-w-f-feasibility} consists of dynamics constraints \eqref{subeq:bimpc-dynamics} and state and input constraints \eqref{subeq:bimpc-stage-constraints}, \eqref{subeq:bimpc-terminal-constraints}.
The input constraint in \eqref{subeq:bimpc-stage-constraints} can be directly written as $u_k \in \mathset{U}, \ k \in \langle \horizon \rangle$.

Given $(x_0, u, \hat{w})$, the dynamics constraint \eqref{subeq:bimpc-dynamics} can be expressed in batch form as follows \cite[Sec.~8.2]{borrelli2017predictive}:
\begin{equation*}
    x = \bar{A}^u x_0 + \bar{B}^u_1 u + \bar{B}^u_2 \hat{w},
\end{equation*}
where $\bar{A}^u \in \real^{\Nx(\horizon+1) \times \Nx}$, $\bar{B}^u_1 \in \real^{\Nx(\horizon+1) \times \Nu\horizon}$, and $\bar{B}^u_2 \in \real^{\Nx(\horizon+1) \times \Nw\horizon}$ are the batch matrices constructed from $A^u$, $B^u_1$, and $B^u_2$.
The robustness constraint \eqref{subeq:robust-bimpc-w-f-feasibility} requires that $C x \leq d$ is satisfied for all $x$ of the form
\begin{equation*}
    x = \bar{A}^u x_0 + \bar{B}^u_1 u + \bar{B}^u_2 \hat{w} + \bar{B}^u_2 \tilde{w},
\end{equation*}
where $\lVert \tilde{w}\rVert \leq \bar{\theta}/\constsc$.
This is equivalent to
\begin{equation*}
    \sup_{\tilde{w}: \lVert \tilde{w}\rVert \leq \bar{\theta}/\constsc} C_j (\bar{A}^u x_0 + \bar{B}^u_1 u + \bar{B}^u_2 \hat{w} + \bar{B}^u_2 \tilde{w}) \leq d_j, \quad \forall j
\end{equation*}
where $C_j, d_j$ are the rows of $C$ and $d$ (see \cite[Sec.~6.4]{boyd2004convex}).
Thus, the robustness constraint can be written as
\begin{equation*}
    C_j (\bar{A}^u x_0 + \bar{B}^u_1 u + \bar{B}^u_2 \hat{w}) + (\bar{\theta}/\constsc)\lVert C_j \bar{B}^u_2\rVert_* \leq d_j, \quad \forall j,
\end{equation*}
where $\lVert \cdot \rVert_*$ is the dual norm~\cite[App.~A.1.6]{boyd2004convex}.
Combining all the above constraints, the constraints of the robust MPC \eqref{eq:robust-bimpc} can be written as
\begin{subequations}
\begin{align}
    & \hat{w}_k \in \mathset{W}, \ u_k \in \mathset{U}, \quad k \in \langle \horizon \rangle, \\
    & \hat{x} = \bar{A}^u x_0 + \bar{B}^u_1 u + \bar{B}^u_2 \hat{w}, \\
    & C_j \hat{x} + (\bar{\theta}/\constsc)\lVert C_j \bar{B}^u_2\rVert_* \leq d_j, \quad \forall j.
\end{align}
\end{subequations}
If some state constraint sets $\mathset{X}$ and $\mathset{X}_\Omega$ in \eqref{eq:bimpc} are given as quadratic constraints, the S-procedure can be used to reformulate the robustness constraint \eqref{subeq:robust-bimpc-w-f-feasibility} \cite[App.~B]{boyd2004convex}.

\bibliographystyle{IEEEtran}
\bibliography{references}

\end{document}